\documentclass[11pt,aps,pre,nofootinbib,bibnotes,notitlepage,showkeys,a4paper,longbibliography]{revtex4-1}
\usepackage{microtype}
\usepackage{graphicx}
\usepackage{amsmath}
\usepackage{amssymb}
\usepackage{amsthm}
\usepackage{xcolor}
\definecolor{myblue}{rgb}{0.153,0.322,0.706}
\usepackage[colorlinks,linkcolor=myblue,urlcolor=myblue,citecolor=myblue]{hyperref}

\newtheoremstyle{myplain}
{5pt}			
{5pt}			
{\normalsize}	
{}			
{\bfseries}		
{:}			
{.5em}		
{\thmname{#1}\thmnumber{ #2}\thmnote{~{(#3)}}}

\theoremstyle{plain}
\newtheorem{theorem}{Theorem}
\newtheorem{proposition}[theorem]{Proposition}
\newtheorem{lemma}[theorem]{Lemma}

\newtheorem{assumption}{Assumption}

\theoremstyle{myplain}
\newtheorem*{remark*}{Remark}
\newtheorem{example}{Example}

\newcommand{\be}{\begin{equation}}
\newcommand{\ee}{\end{equation}}
\newcommand{\ra}{\rightarrow}
\newcommand{\cM}{\mathcal{M}}
\newcommand{\cB}{\mathcal{B}}
\newcommand{\cA}{\mathcal{A}}
\newcommand{\reals}{\mathbb{R}}
\newcommand{\treals}{\bar\reals}
\newcommand{\p}{\partial}
\newcommand{\bX}{\mathbf{X}}
\newcommand{\tp}{\tilde p}
\newcommand{\hp}{\hat p}
\newcommand{\idf}{\mathbf{1}}
\newcommand{\EX}{\mathbb{E}}
\newcommand{\var}{\textrm{Var}}
\newcommand{\cN}{\mathcal{N}}
\newcommand{\cE}{\mathcal{E}}

\newcommand{\iid}{i.i.d.}
\newcommand{\eex}{\hfill$\blacksquare$}
\newcommand{\hm}{b}
\newcommand{\barm}{\bar m}
\newcommand{\bB}{\bar B}

\begin{document}

\title{Efficient large deviation estimation based on importance sampling}

\author{Arnaud Guyader}
\email{arnaud.guyader@upmc.fr}
\affiliation{\mbox{Laboratoire de Probabilit\'es, Statistique et Mod\'elisation, Sorbonne Universit\'e, Paris, France}}
\affiliation{CERMICS, \'Ecole des Ponts ParisTech, France}

\author{Hugo Touchette}
\email{htouchet@alum.mit.edu, htouchette@sun.ac.za}
\affiliation{Department of Mathematical Sciences, Stellenbosch University, Stellenbosch, South Africa}

\date{\today}

\begin{abstract}
We present a complete framework for determining the asymptotic (or logarithmic) efficiency of estimators of large deviation probabilities and rate functions based on importance sampling. The framework relies on the idea that importance sampling in that context is fully characterized by the joint large deviations of two random variables: the observable defining the large deviation probability of interest and the likelihood factor (or Radon--Nikodym derivative) connecting the original process and the modified process used in importance sampling. We recover with this framework known results about the asymptotic efficiency of the exponential tilting and obtain new necessary and sufficient conditions for a general change of process to be asymptotically efficient. This allows us to construct new examples of efficient estimators for sample means of random variables that do not have the exponential tilting form. Other examples involving Markov chains and diffusions are presented to illustrate our results.
\end{abstract}


\maketitle

\section{Introduction}
\label{secintro}

Estimating the probability of rare events or fluctuations in random systems is an important problem arising in many applied fields, including engineering \cite{shwartz1995}, where a rare event might represent a design failure, or chemistry, where changes between chemical species or polymer states arise from rare transitions in a free energy landscape \cite{wales2004,eijnden2004,lelievre2010}. In physical systems, the probability of rare fluctuations often has a large deviation form \cite{ellis1985,dembo1998,hollander2000,touchette2009}, owing to the interaction of many particles or the effect of thermal noise. In this case, the estimation of probabilities reduces to the estimation of \emph{rate functions}, which determine the rate of decay of  probabilities as a function of some parameter (e.g., volume, particle number, integration time or temperature) \cite{touchette2009}. 

Rate functions are also important on their own, as they determine for equilibrium and nonequilibrium systems the onset of static and dynamical phase transitions \cite{touchette2009,garrahan2007,garrahan2010,espigares2013,bunin2013,tsobgni2016b,lazarescu2017}, fluctuation symmetries \cite{gallavotti1995,kurchan1998,lebowitz1999,harris2007}, and in some cases the response to external perturbations \cite{baiesi2009}. As a result, they have been actively studied recently, especially for nonequilibrium systems describing particle transport processes \cite{derrida2007,bertini2007,harris2013,garrahan2018} and diffusing particles \cite{sekimoto2010,seifert2012,seifert2018,ciliberto2017}, among other physical systems.

Traditionally, two statistical methods have been used to numerically estimate or sample large deviation probabilities: 1) \emph{splitting} \cite{cerou2007,dean2009,cerou2011,cerou2019,cerou2019b,brehier2019}, also known as cloning in physics \cite{grassberger2002,giardina2006,lecomte2007a,angeli2019}, which works by replicating events that ``go in the direction'' of the rare event of interest, and 2) \emph{umbrella} or \emph{importance sampling} (IS) \cite{torrie1977,juneja2006,asmussen2007,bucklew2004}, which works by modifying the process simulated so as to increase the likelihood of the event of interest and, ideally, to render it typical. The probability of that event is then computed via the likelihood factor or Radon--Nikodym derivative, which is the bridge connecting probabilities in the original and the modified processes.

In this paper, we consider the latter method with the aim of providing a complete framework for understanding the efficiency of IS when used to estimate large deviation probabilities and rate functions. For this purpose, we first review in Sec.~\ref{secis} the basis of IS as applied to large deviation estimation, and then present the main results known about the efficiency of IS, which we illustrate with simple examples involving sums of random variables.

Most of these results were obtained by Bucklew and Sadowsky \cite{bucklew2004,sadowsky1989,sadowsky1990,bucklew1990b} (see also \cite{schlebusch1993,dieker2005,efron1968}) and are based on two basic but important observations. The first, found in any presentation of IS, is that, although it is necessary for an efficient change of process or ``measure'' in IS to render rare events typical, this is not sufficient, as we must also ensure that the IS probability estimator arising from the change of process has good variance properties \cite{asmussen2007}. The second observation, which is specific to large deviations, is that the notion of a ``good'' or an ``efficient'' change of process must be adapted to the exponentially decaying form of probabilities that we are trying to estimate. Thus, instead of seeking changes of process that achieve zero variance or a bounded relative error, which are too prohibitive, we must look for changes of process whose second moment decays exponentially with the largest rate possible \cite{bucklew2004}. This leads to the notion of \emph{logarithmic efficiency} or \emph{asymptotic efficiency}, defined in a precise way in the next section.

Following this review part of the paper, we present in Sec.~\ref{secresults} a new framework for determining and understanding whether a change of process is asymptotically efficient or not. The framework is itself based on large deviation theory and draws on the idea, recently put forward by one of us \cite{touchette2018}, that changes of processes and measures in general are completely characterized in the context of large deviation probabilities by the joint rate function of two random variables: 1) the random variable defining the rare event of interest, and 2) the Radon--Nikodym derivative, seen as a real random variable with respect to either the original or the modified process. 

The resulting framework recovers results previously known about the efficiency of IS for large deviation estimation \cite{bucklew2004,sadowsky1989,sadowsky1990,bucklew1990b,schlebusch1993}, but also extends them in two important ways. First, most of the results that have been derived in the past and that are now used in practice apply to a specific change of process known as the exponential tilting, the exponential family or the Esscher transform. By contrast, our formalism can be applied in principle to any change of process to determine whether that change is efficient and, if not, to understand in a clear way why this is so. Second, most works provide sufficient but not necessary conditions for asymptotic efficiency. For the exponential tilting, these conditions are based on the existence of so-called dominating points, related essentially to the convexity of rate functions and the convexity of the rare event set. They can be checked in many applications of interest, leading to efficient IS simulations, but they leave completely open the possibility that changes of process other than the exponential tilting can be asymptotically efficient. Indeed, the full characterization of such changes is still an open problem in IS as applied to large deviation estimation.

Here, we solve this problem by providing in Sec.~\ref{secresults} necessary and sufficient conditions for a change of process to be asymptotically efficient. We use these conditions in  Sec.~\ref{secexamples} to revisit the efficiency of the exponential tilting, and then illustrate them with explicit examples of large deviations involving independent random variables and discrete-time Markov chains. From these, we also construct two intriguing examples of IS estimators that do not have the exponential tilting form and yet are asymptotically efficient, opening the way for more to be discovered. Applications to stochastic differential equations are finally presented to illustrate how our results can be applied beyond discrete-time models to estimate the large deviations of continuous-time Markov processes, commonly used as models of nonequilibrium systems.

\section{Importance sampling of large deviations}
\label{secis}

We define in this section the rare event or large deviation probabilities that we are interested in estimating using importance sampling and define the notion of asymptotic efficiency, used classically in the context of large deviations. Most of the results reviewed can be found in Bucklew's book \cite{bucklew2004}, which follows the prior works~\cite{sadowsky1989,sadowsky1990,bucklew1990b}.

\subsection{Large deviation probabilities}

The rare events that we consider are defined in a general way by considering two ingredients:
\begin{itemize}
\item A sequence $\bX_n=(X_1,X_2,\ldots,X_n)$ of random variables taking values on some space $\Lambda_n$ with probability measure $P_n$;
\item A function $M_n: \Lambda_n \ra \cM$, referred to as an \emph{observable}.
\end{itemize}

Concretely, $\bX_n$ represents the state of some system or process, $P_n$ is the probabilistic model (the prior measure) that we have of that system, while $M_n$ is some function of that system that can be observed or measured in some way. For example, $\bX_n$ can be the microstate of an equilibrium system of $n$ classical particles, in which case $P_n$ is the ensemble (microcanonical, canonical, etc.) chosen to ``weight'' the microstates and $M_n$ can represent the particles' energy. The system can also be a stochastic process, e.g., a Markov chain in discrete time, with $\bX_n$ representing its path or trajectory over $n$ time steps, $P_n$ the probability measure over the trajectories, which defines the process, and $M_n$ a function of the trajectories.

For simplicity, we consider the case where $X_i\in \reals^d$, $d\geq 1$, so that $\Lambda_n= (\reals^d)^n$, and $M_n(\bX_n)\in\reals^D$, $D\geq 1$, so that $\cM\subset\reals^D$. More general spaces can be used for both the process and the observable at the expense of introducing more complicated notations. For example, it is common in large deviation theory to consider $\cM$ to be a Polish space to handle cases where $M_n$ takes values in a function space, e.g., if $M_n$ is an empirical distribution, as in Sanov's theorem \cite{dembo1998}. 

Here, we limit ourselves to a setting where both $\bX_n$ and $M_n$ are finite-dimensional random variables, so as to simplify the notations. In fact, most of our results will be illustrated by considering simple examples where $M_n$ is a sample mean of real random variables
\be
M_n = \frac{1}{n}\sum_{i=1}^n X_i,
\ee
so that both $X_i\in\reals$ and $M_n\in\reals$. From these, it is easy to generalize to other processes and observables, including observables defined for Markov chains or even continuous-time Markov processes, as shown in Sec.~\ref{secexamples}.

Given $\bX_n$, $P_n$ and $M_n$, we are interested in estimating the probability
\be
p_n\equiv P_n(M_n\in B),
\label{eqldprob1}
\ee
where $B$ is some measurable subset of $\cM$ and $P_n$ denotes, with a slight abuse of notation, the probability measure extended to $M_n$. As a particular case, we can set  $B=[m,m+d m]$ to obtain, as is common in physics, the probability distribution of $M_n$ with ``discretization'' $d m$. Our basic assumption is that this probability has a large deviation form with $n$, meaning that it decays exponentially with $n$ and so describes a rare event that becomes rarer as $n$ gets larger. 

This decay of probabilities is encountered in many applications and can be expressed mathematically in different ways, depending on the level of generality adopted. Here, we say that $P_n(M_n\in B)$ has a large deviation form or satisfies, more precisely, the \emph{large deviation principle} (LDP) if there exists a function $I_P:\cM\ra[0,\infty]$ such that
\be
\lim_{n\ra\infty} -\frac{1}{n}\log P_n(M_n\in B)= I_P(B),
\label{eqldp1}
\ee
where
\be
I_P(B)=\inf_{m\in B} I_P(m).
\label{eqldpmin1}
\ee
The function $I_P$ is called the \emph{rate function} of $M_n$ and is required to be lower semi-continuous, meaning that it has closed level sets. We assume, as is common in large deviation theory, that $I_P$ is in fact a \emph{good} rate function, meaning that it has compact level sets. This simplifies the analysis of large deviations, as it implies that the infimum in \eqref{eqldpmin1} is attained on at least one point in the closure $\bar B$ of $B$ \cite[Sec.~1.2]{dembo1998}. It also means for $M_n\in\reals^D$ that $I_P$ is \emph{coercive}, that is, $I_P(m)\ra\infty$ as $\|m\|\ra\infty$. The first assumption of our work is thus:

\begin{assumption}\label{apzicjicj}
\label{hyp1} The observable $M_n$ satisfies the LDP, in the sense of \eqref{eqldp1}, with good rate function $I_P$ such that $I_P(B)<\infty$.
\end{assumption}

The limit in \eqref{eqldp1} is actually a simplification of the standard definition of the LDP found in the large deviation literature involving upper and lower bounds (see, e.g., \cite[Sec.~1.2]{dembo1998}). In using the definition above, we assume that $B$ is a ``good'' set, called technically an $I$-continuity set \cite[Sec.~1.2]{dembo1998}, such that 
\be
\inf_{m\in B} I_P(m) = \inf_{m\in B^\circ} I_P(m) = \inf_{m\in \bar B} I_P(m),
\ee
where $B^\circ$ represents the interior of $B$. In this case, the upper and lower bounds appearing in the standard definition of the LDP are the same, yielding the simple limit \eqref{eqldp1}. This is a technical point, which is not important for physical or numerical applications. 

Concretely, the LDP means again that the leading behavior of the distribution of $M_n$ is a decaying exponential in $n$, with corrections in the exponential that are smaller than linear in $n$. This property is commonly summarized in large deviation theory by the asymptotic notation \cite{ellis1985,dembo1998,hollander2000,touchette2009} 
\be
P_n(M_n\in [m,m+dm])\asymp e^{-n I_P(m)},
\ee
and applies whenever $I_P(m)>0$. When $I_P(m)=0$, the distribution of $M_n$ either decays around $m$ slower than exponentially in $n$ or increases with $n$ around that point. In many applications, $I_P(m)$ has only one zero, denoted in the following by $m^*$, so the latter case applies, yielding the law of large numbers
\be
\lim_{n\ra\infty} P_n(M_n \in [m^*,m^*+dm])=1
\ee
or, more generally,
\be
\lim_{n\ra\infty} P_n(M_n\in B)=1
\ee
if $m^*\in B^\circ$. See \cite{touchette2009} and references therein for cases where more than one zeros occur.

Probabilities having the LDP form are encountered in many applications of interest, including queues \cite{shwartz1995}, hypothesis testing \cite{hollander2000}, and noisy detection systems \cite{dembo1998}. In physics, the LDP is the basis of thermodynamics and describes, more generally, the fluctuations of equilibrium systems in the thermodynamic limit of large systems, which is a large deviation limit (see \cite{touchette2009} for a review). The same exponential form of probabilities also arises in the context of nonequilibrium systems when considering systems perturbed by a small noise \cite{cottrell1983,freidlin1984,graham1989,luchinsky1998} as well as time-integrated functions or observables of Markov processes \cite{touchette2017} modelling, for example, the fluctuating dynamics of mesoscopic diffusive systems \cite{sekimoto2010,seifert2012,seifert2018} or many-particle transport processes \cite{derrida2007,bertini2007,harris2013}. In the latter case, the long-time limit is often combined with a low-noise limit describing the residual noise associated with a macroscopic (or hydrodynamic) limit where infinitely many interacting particles evolve in time over a substrate (e.g., a lattice) with boundary reservoirs \cite{bertini2015b}. 

\subsection{Importance sampling}

The simplest way to numerically estimate $p_n$ in (\ref{eqldprob1}) is to sample $M_n$ directly by generating multiple copies $\bX_n^{(i)}$, $i=1,2,\ldots,N$, of the state from the probability measure $P_n$ and by then counting the fraction of corresponding observable values $M_n^{(i)}=M_n(\bX_n^{(i)})$ that fall in $B$:
\be
\tp_n^N \equiv \frac{1}{N}\sum_{i=1}^N \idf_{M_n^{(i)}\in B}.
\ee
Since the random variables $\idf_{M_n^{(i)}\in B}$ are independent and identically distributed (\iid) Bernoulli with parameter $p_n$, it is easy to see that the estimator above, referred to as the direct or \emph{crude Monte Carlo} (CMC) estimator, is unbiased in the sense that
\be
\EX_P[\tp_n^N] = p_n,
\ee
where $\EX_P[\cdot]$ denotes the expectation with respect to $P_n$. Moreover, its variance is
\be
\var_P(\tp_n^N)= \EX_P[(\tp_n^N -p_n)^2] = \frac{p_n(1-p_n)}{N}
\ee
and so decreases with $N$. However, since $p_n$ becomes exponentially small as $n\ra\infty$, the actual number of samples needed to accurately estimate this probability should be determined from the estimator's error relative to $p_n$, which can be approximated by
\be
\frac{\sqrt{\var_P(\tp_n^N)}}{p_n}\approx \frac{1}{\sqrt{N p_n}}.
\ee
As a result, we see that $N$ must grow exponentially as $N\sim p_n^{-1} \asymp e^{nI_P(B)}$ in order for the relative error to be bounded in $n$, which is unachievable in practical simulations.

To overcome this problem, we resort to IS which works by sampling $M_n$ not according to $P_n$ but to a different probability measure $Q_n$, chosen to increase the likelihood that $M_n\in B$ \cite{juneja2006,asmussen2007,bucklew2004}. To be consistent, $Q_n$ must have support on all states that ``hit'' the event $\{M_n\in B\}$ with respect to $P_n$, which translates mathematically to requiring that $P_n\idf_{M_n\in B}$, the restriction of $P_n$ on $\{M_n\in B\}$, be absolutely continuous with respect to $Q_n\idf_{M_n\in B}$ \cite{asmussen2007}. Here, we assume for simplicity that $Q_n$ has the same support as $P_n$, so the two are equivalent in the sense of absolute continuity.

To estimate $p_n$, we now generate copies $\bX_n^{(i)}$, $i=1,2,\ldots,N$, of the states according to $Q_n$,\footnote{We could identify the new copies with a different symbol, say $\tilde\bX_n^{(i)}$, since they are generated from a different distribution and so represent a different random variable. Here, we keep $\bX_n^{(i)}$ but always specify the distribution, $P_n$ or $Q_n$, used. The same applies to the observable.} compute the associated observable values $M_n^{(i)}$, $i=1,2,\ldots,N$, and construct the IS estimator as
\be
\hp_n^N \equiv\frac{1}{N}\sum_{i=1}^N L_n^{(i)}\, \idf_{M_n^{(i)}\in B},
\label{eqisest1}
\ee
where $L_n^{(i)}=L_n(\bX_n^{(i)})$ and
\be
L_n \equiv \frac{dP_n}{dQ_n}
\ee
is the \emph{Radon--Nikodym derivative} of $P_n$ with respect to $Q_n$. This derivative, also known as the likelihood factor, is included to ensure that the IS estimator remains unbiased, that is,
\be
\EX_Q[\hp_n^N] = \int_{\Lambda_n} dQ_n(\bX_n)\, \frac{dP_n}{dQ_n}(\bX_n)\, \idf_{M_n(\bX_n)\in B}=\EX_P[\idf_{M_n\in B}] = p_n.
\label{equnb1}
\ee
The variance, however, is modified to
\be
\var_Q(\hp_n^N) = \EX_Q[(\hp_n^N -p_n)^2] =\frac{\EX_Q[L_n^2 \idf_{M_n\in B}]-p_n^2}{N},
\label{eqisvar1}
\ee
which obviously depends on $Q_n$, leading to the relative variance
\be
\frac{\var_Q(\hp_n^N)}{p_n^2}=\frac{\EX_Q[L_n^2 \idf_{M_n\in B}]}{Np_n^2}-\frac{1}{N}.
\label{eqisrelvar1}
\ee

The problem of IS is to determine which $Q_n$ is \emph{optimal}, that is, which choice achieves the smallest variance or relative variance, ideally smaller than the CMC variance obtained with $Q_n=P_n$, given some design or application-specific constraints on the class of $Q_n$ that can be simulated. 

If no constraints are imposed, then it is known that there is a zero-variance change of measure given by the restriction of $P_n$ on the event of interest, that is, $Q_n\propto P_n\idf_{M_n\in B}$. This measure, known in physics as the microcanonical ensemble \cite{touchette2015}, cannot be simulated in practice, however, because it involves a hard-to-implement constraint and, more importantly, because its normalization involves the probability that we are trying to estimate. As a result, other choices must be considered that are either approximations of the zero-variance solution (following, e.g., cross-entropy methods \cite{rubinstein2004}) or that are optimal or efficient according to some bounding criterion on the variance or relative variance. For the purpose of estimating large deviations, a common criterion used is the \emph{asymptotic (or logarithmic) efficiency} \cite{bucklew2004,juneja2006,asmussen2007}, discussed next. 

\subsection{Asymptotic efficiency}

The notion of asymptotic efficiency is based on the relative variance of the IS estimator, as given by \eqref{eqisrelvar1}. Since $p_n$ scales exponentially with $n$, so does generally the second moment $\EX_Q[L_n^2 \idf_{M_n\in B}]$ with a scaling exponent given by
\be
R_Q(B) \equiv\lim_{n\ra\infty} -\frac{1}{n} \log \EX_Q[L_n^2 \idf_{M_n\in B}].
\label{eqde1}
\ee
We will provide in Section \ref{lkzx} specific assumptions that ensure the existence of this limit (see Lemma \ref{lem1}). For now, we note that the LDP for $p_n$, combined with the positivity of the variance in \eqref{eqisvar1}, implies
\be
R_Q(B)\leq 2I_P(B).
\label{eqeff1}
\ee
When equality is achieved, we say that the IS measure $Q_n$ or, more precisely, the sequence $(Q_n)_{n>0}$ of IS measures, is \emph{asymptotically efficient}. This criterion is also referred to in the literature as \emph{logarithmic efficiency} or \emph{asymptotic optimality}.

It can be checked that CMC achieves $R_Q(B) = I_P(B)$ and so it is not asymptotically efficient, as expected, while the zero-variance choice $Q_n\propto P_n\idf_{M_n\in B}$ is asymptotically efficient, since it has zero variance for all $n$. By comparison, an asymptotically efficient $Q_n$ does not necessarily have zero variance -- this is again too restrictive for our purpose -- but is such that the term $\EX_Q[L_n^2 \idf_{M_n\in B}]$ in the variance decays with the fastest exponential rate equal to $2I_P(B)$. When this happens, the ratio $\EX_Q[L_n^2 \idf_{M_n\in B}]/p_n^2$ in \eqref{eqisrelvar1} does not grow exponentially with $n$, which means that the number $N_n$ of samples needed to have a fixed relative variance grows sub-exponentially in $n$. Hence, if $Q_n$ is asymptotically efficient, we have
\be
\lim_{n\ra\infty}\frac{1}{n}\log N_n =0.
\label{eqaelim1}
\ee
This is often taken as a definition of asymptotic efficiency.

The asymptotic efficiency of $Q_n$ is studied in most works \cite{bucklew2004,sadowsky1989,sadowsky1990,bucklew1990b,schlebusch1993,dieker2005} for the \emph{exponential tilting} or \emph{exponential family}, defined by
\be
Q_{n}(d\bX_n) = \frac{e^{n\langle k, M_n\rangle} P_n(d\bX_n)}{\EX_P[e^{n\langle k,M_n\rangle}]},
\label{eqexptilt1}
\ee
where $k\in\reals^D$ is a vector having the same dimension as $M_n$ and $\langle \cdot,\cdot\rangle$ is the standard scalar product in $\reals^D$. We will not review all the results known about this change of measure, which also corresponds in physics to the canonical ensemble \cite{touchette2015}. For our purposes, two results are worth noting. The first, proved in \cite{sadowsky1990}, states that the exponential tilting is asymptotically efficient if $B$ has a dominating point (see \cite[Sec.~5.2]{bucklew2004} for a definition of this concept), which holds essentially when $I_P(m)$ is a convex function and $B$ is a convex set. In that case, the value $k\in\reals^D$ that must be chosen in \eqref{eqexptilt1} to achieve efficiency satisfies
\be
\nabla \lambda_P(k)=\mu,
\label{eqdompt1}
\ee 
where $\mu\in\reals^D$ is the dominating point of $B$ and $\lambda_P(k)$ is the \emph{scaled cumulant generating function} (SCGF), defined by
\be
\lambda_P(k)=\lim_{n\ra\infty} \frac{1}{n} \log \EX_P[e^{n\langle k, M_n\rangle}],\qquad k\in\reals^D.
\label{eqscgf1}
\ee

We refer to \cite[Thm.~2]{sadowsky1990} for the complete statement of this result, including the conditions underlying it, and \cite[Thm.~3]{bucklew1990b} for its application to Markov chains. We give some examples next to illustrate the relation \eqref{eqdompt1}, which comes from the application of the G\"artner--Ellis theorem and the fact, more precisely, that the rate function is given, according to this theorem, by the Legendre--Fenchel transform of the SCGF when the latter is differentiable; see \cite[Sec.~4.4]{touchette2009} for more details. In fact, the conditions underlying the efficiency of the exponential tilting are overall nothing but the conditions of the G\"artner--Ellis theorem. 

The second result worth noting, also found in \cite{sadowsky1990}, is that the sample size $N_n$ required for the IS estimator $\hp_n^N$ to have a bounded relative variance  grows according to
\be
N_n\asymp e^{n[2I_P(B)-R_Q(B)]}
\label{eqsamplesize1}
\ee
in the limit $n\ra\infty$. This essentially follows from the result \eqref{eqisrelvar1} for the relative variance of the IS estimator, in which the term $1/N$ can be neglected. In particular, $N_n\asymp e^{nI_P(B)}$ for CMC, as seen before, while $N_n\asymp e^{n0}$ if $Q_n$ is asymptotically efficient, consistently with the limit \eqref{eqaelim1} above and the fact again that constant relative variance is achieved in this case by increasing $N_n$ sub-exponentially with $n$. In some cases, it turns out in fact that bounded relative variance is achieved with $N_n=O(\sqrt{n})$ when $Q_n$ is asymptotically efficient \cite[Sec.~5.4]{bucklew1990b}, leading to a drastic increase in simulation efficiency.

\subsection{Examples}

We close this section by illustrating the theory developed so far with simple examples involving the sample mean
\be
M_n=\frac{1}{n}\sum_{i=1}^n X_i
\ee
of a sequence $\bX_n=(X_1,X_2,\ldots, X_n)$ of \iid\ random variables. The examples are presented briefly, since they appear in  other works (see, e.g., \cite{bucklew2004}), and will be used again in the next sections to illustrate our new framework. More involved applications of IS related to large deviations have been considered in the context of random graphs \cite{engel2004,hartmann2011b,dewenter2015}, finance \cite{guasoni2008}, escape pathways \cite{eijnden2012}, and nonequilibrium systems \cite{kundu2011,klymko2018,whitelam2018c,jacobson2019}, among other topics.

\begin{example}
\label{exgauss1}
We first consider a sample mean of standard Gaussian random variables, so that $P_n$ is the product measure $\cN(0,1)^{\otimes n}$, and look for the probability that $p_n=P_n(M_n\geq 1)$ by choosing $B=[1,\infty)$. This probability can be found directly from the fact that $M_n\sim\cN(0,1/n)$, leading to $p_n\asymp e^{-n/2}$. Alternatively, we can use Cram\'er's theorem \cite[Sec.~2.2]{dembo1998} to find that the SCGF is $\lambda_P(k)=k^2/2$, yielding $I_P(m) = m^2/2$ by Legendre--Fenchel transform and, therefore,
\be
I_P(B) = \inf_{m\geq 1} I_P(m) = \frac{1}{2}.
\label{eqgaussex1}
\ee
This shows that the probability $P_n(M_n\geq 1)$ is dominated exponentially by the probability that $M_n$ is close to $1$, so only the boundary of $B$ plays a role, as is common with large deviations.

With this result, it is natural to choose the IS measure to be a sequence of \iid\ Gaussian random variables centered at $1$, so that $Q_n=\cN(1,1)^{\otimes n}$. It is clear that this change of measure makes $M_n=1$ typical. Moreover, it can be checked by calculating $R_Q(B)$ directly from its definition \eqref{eqde1} that this measure is asymptotically efficient with
\be
R_Q(B)=1=2I_P(B).
\ee
Alternatively, one can notice, following \cite[Ex.~5.2.1]{bucklew2004}, that the dominating point of $B=[1,\infty)$ is $\mu=1$, which leads, with the equation $\lambda_P'(k)=\mu$, to $k=1$. From \eqref{eqexptilt1}, the exponentially-tilted measure that is asymptotically efficient is then
\be
Q_n(d\mathbf{X}_n)=\frac{e^{nM_n}P_n(d\mathbf{X}_n)}{\EX_P[e^{nM_n}]}
=\left\{\prod_{i=1}^{n}\frac{e^{-(X_i-1)^2/2}}{\sqrt{2\pi}}dX_i\right\},
\ee
which is indeed the product measure $\mathcal{N}(1,1)^{\otimes n}$. Note that the Radon--Nikodym derivative of $P_n$ with respect to $Q_n$ is
\be
L_n=L_n(\mathbf{X}_n)=\frac{dP_n}{dQ_n}(\mathbf{X}_n)=e^{-n(M_n-1/2)}.
\ee
Therefore, in the end, the IS estimator that is asymptotically efficient is
\be
\hat{p}_n^N=\frac{1}{N}\sum_{i=1}^{N}e^{-n(M_n^{(i)}-1/2)}\mathbf{1}_{M_n^{(i)}\geq 1},
\label{eqisg1}
\ee
where $\{M_n^{(i)}\}_{i=1}^N$ is an \iid\ sample generated with $Q_n$.\eex
\end{example}

This example can be generalized, obviously, to any $B=[b,\infty)$, $b>0$, by choosing $k$ in the exponential tilting such that $\lambda_P'(k)=b$ or, equivalently, $k=I'(b)$ by Legendre duality (see Sec.~3.5 of \cite{touchette2009}). This gives $Q_n=\cN(b,1)^{\otimes n}$ as the modified measure that changes the event $\{M_n\geq b\}$ from being rare to being typical. This is asymptotically efficient, as $b$ is the dominating point of $B$. Choosing $Q_n$ to concentrate \emph{inside} $B$ rather than at its boundary, that is, $Q_n=\cN(b',1)^{\otimes n}$ with $b'>b$, is not asymptotically efficient, although it does make $\{M_n\geq b\}$ typical.

As a variation of this example, we change the distribution of the $X_i$'s to an exponential distribution. Other distributions, such as Bernoulli, uniform or Laplace, are treated in \cite{bucklew2004}.

\begin{example}
\label{exexpsm1}
Let the sequence $X_1,X_2,\ldots,X_n$ of \iid\ random variables be distributed according to the exponential distribution $\cE(1)$ with parameter $1$, so that $P_n=\cE(1)^{\otimes n}$. We consider again the sample mean $M_n$ as an observable and $B=[b,\infty)$ with $b>1$, so that $p_n=P_n(M_n\in B)$ is a rare event such that \cite{touchette2009}
\be
I_P(B)= b-1-\log b.
\label{exrateexp1}
\ee
As shown in \cite[Ex.~5.2.6]{bucklew2004}, the asymptotically efficient exponential tilting associated with this problem is the product measure $Q_n=\mathcal{E}(1/b)^{\otimes n}$ of exponential distributions with mean $\EX_Q[X_i] = \EX_Q[M_n]=b$. This follows by noting that $\lambda_P(k)=-\log(1-k)$ for $k<1$, from which we find $k=1-1/b$ by solving $\lambda_P'(k)=b$. Equivalently, $k=I'(b) = 1-1/b$.\eex
\end{example}

The last example is a classic one in IS showing that the exponential tilting is not always asymptotically efficient, in particular, when dealing with nonconvex sets $B$.

\begin{example}
\label{exnonconv1}
Consider, as in the first example, a sequence of \iid\ normal random variables with the same $P_n$ and $Q_n$, but now take $B$ to be the union of two disjoint sets, namely, $B=(-\infty,-b]\cup [1,\infty)$ with $b>1$, so that the probability to estimate is
\be
p_n = P_n(M_n\leq -b \textrm{ or } M_n\geq 1).
\ee
From Example~\ref{exgauss1}, we still have $I_P(B)=1/2$, since $M_n\leq -b$ is rarer than $M_n\geq 1$ for $b>1$. The calculation of $R_Q(B)$ for this case can be found in~\cite[Ex.~5.2.13]{bucklew2004}. The result is $R_Q(B)=1$ if $b\geq 3$ and $R_Q(B) <1$ otherwise, implying that $Q_n=\cN(1,1)^{\otimes n}$ is not asymptotically efficient if $b\in (1,3)$. Note that this cannot be inferred from the dominating point result, since $B$ is nonconvex and, as such, has no dominating point for any $b$. \eex
\end{example}

The non-efficiency of $Q_n$ in the last example is due to the fact that, although the probability that $M_n \leq -b$ is exponentially small, this rare event leads to exponentially large values of the likelihood factor in \eqref{eqisg1} that dramatically increase the variance of the IS estimator. In fact, it can be checked (see again \cite[Ex.~5.2.13]{bucklew2004}) that for values of $b$ close to $-1$, $R_Q(B)$ becomes negative, so that the second moment of the IS estimator can diverge with $n$ as a result of the accumulation of many different likelihood factors that are exponentially large with $n$.

Other examples involving nonconvex sets $B$ have been studied, in particular, by Glasserman and Wang \cite{glasserman1997}, who show that IS based on the exponential tilting can perform worse than CMC and can even lead to $R_Q(B)=-\infty$, so one should be cautious about the generally-accepted idea that a good choice of IS measure is one that makes a rare event typical. 

To avoid the case where $R_Q(B)=-\infty$, which is clearly not efficient, we assume here that $R_Q(B)>-\infty$. In fact, for the results to come, we need a slightly stronger assumption:

\begin{assumption}
\label{hyp2} For some $\delta>0$,
\be
\limsup_{n\ra\infty} \frac{1}{n} \log \EX_Q[L_n^{2+\delta} \idf_{M_n\in \bar B}] < \infty.
\label{eqass2}
\ee
\end{assumption}

This condition implies with H\"older's inequality that, if $R_Q(B)$ exists, then $R_Q(B)>-\infty$ and, therefore, with Assumption~\ref{hyp1} and \eqref{eqeff1}, that $R_Q(B)$ is finite. It also ensures overall that we are not in a situation where the second moment of $L_n$ has the correct exponential scaling in $n$, but its moment of order $2+\delta$ behaves super-exponentially. 

This type of ``Lyapunov'' condition often appears in large deviation theory in the context of Varadhan's integral lemma (see \cite[Thm.~4.3.1]{dembo1998}). Other weaker conditions can be defined (see \cite[Thm.~4.3.1]{dembo1998} and \cite{dieker2005} in the context of IS), although they might be more difficult to check.  In our case, we will see in the next section that the limit above can be re-expressed more naturally in terms of the steepness or coercivity of a rate function involving $M_n$ and $L_n$ (see Assumption $2'$).

\section{Joint LDP approach to asymptotic efficiency}
\label{secresults}

The theory presented in the previous section can be applied to sample large deviations in an efficient way not just for \iid\ sample means, as illustrated, but also for functionals of Markov chains, jump processes, and diffusions. One problem of the theory, however, is that it focuses  almost exclusively on the exponential tilting, leaving aside the possibility that other changes of measure might also be asymptotically efficient. Moreover, the efficiency conditions that we have for the exponential tilting, based on the existence of a dominating point, are only sufficient conditions that cannot be applied to nonconvex problems, as illustrated in Example~\ref{exnonconv1}. In principle, one can determine the efficiency of an arbitrary $Q_n$ by calculating the exponent $R_Q(B)$ \cite{bucklew2004}, but this is very difficult to carry out in practice beyond the case of \iid\ sample means and convex $B$.

We address these issues in this section by providing new conditions for a general change of measure $Q_n$ to be asymptotically efficient. These conditions follow from two basic observations about the second moment $\EX_Q[L_n^2 \idf_{M_n\in B}]$ that determines the efficiency of $Q_n$ via \eqref{eqeff1}. The first is that this moment involves both $L_n$ and $M_n$, which means that it can be computed knowing the joint distribution of these two (generally correlated) random variables. The second is that, in many cases of interest, the Radon--Nikodym derivative $L_n$ scales exponentially in $n$ and has a distribution that satisfies the LDP \cite{puhalskii1998}. Therefore, it is natural to study the efficiency of $Q_n$ based on the joint large deviations of $M_n$ and $L_n$, which is what we do in this section. 

By reformulating the asymptotic efficiency criterion in terms of a joint LDP involving $M_n$ and $L_n$, we derive necessary and sufficient conditions for a general $Q_n$ to be asymptotically efficient. These conditions provide new insights into what makes a change of measure efficient. They show, in particular, that $L_n$ does not have to be deterministic conditionally on $M_n$, which is the essential property of the exponential tilting that makes it asymptotically efficient. The fluctuations of $L_n$ only need to be ``bounded'' or ``controlled'' in a precise way, suggesting new changes of measure, different from the exponential tilting, that are asymptotically efficient.

\subsection{Joint large deviations}\label{lkzx}

The idea of formulating a joint LDP for $M_n$ and $L_n$ follows the recent work of one of us \cite{touchette2018}. As in that work, we consider $L_n$ via the scaled log-likelihood or \emph{action}, defined as
\be
W_n \equiv -\frac{1}{n} \log L_n,
\ee
to account for the fact that $L_n$ is expected to scale exponentially with $n$. The action is obviously a real random variable whose distribution can be determined in principle with respect to either $P_n$ or $Q_n$. The couple $(M_n,W_n)$ is thus a random variable taking values in the product space $\cM\times\reals$, where $\cM$, the space of $M_n$, is again a subset of $\reals^D$. 

From now on, we assume the following:
\begin{assumption}~
\label{hyp3}
\begin{enumerate}
\item[\emph{(a)}] $(M_n,W_n)$ satisfies the LDP relative to $P_n$ on $\cM\times\reals$ with good rate function $J_P$;
\item[\emph{(b)}] $(M_n,W_n)$ satisfies the LDP relative to $Q_n$ on $\cM\times\reals$ with good rate function $J_Q$;
\item[\emph{(c)}] $J_P$ and $J_Q$ have the same non-empty domain on $\cM\times\reals$, that is, the same set of values on which these functions are finite.
\item[\emph{(d)}] $B\times\reals$ is a good set for $(m,w)\mapsto 2w+J_Q(m,w)$, meaning that
\be
\inf_{(m,w)\in B\times\reals}\left\{2w+J_Q(m,w)\right\} = \inf_{(m,w)\in B^\circ\times\reals} \left\{2w+J_Q(m,w)\right\} = \inf_{(m,w)\in \bar B\times\reals} \left\{2w+J_Q(m,w)\right\}.
\ee
\end{enumerate}
\end{assumption}

In the last assumption, there is an abuse of language, since the function $2w+J_Q(m,w)$ is not necessarily a rate function. 

The LDPs in Conditions (a)-(b) follow the rigorous definition given in Sec.~\ref{secis} and mean in terms of the asymptotic notation that
\be
P_n(M_n\in B, W_n\in C) \asymp \exp\left\{-n \inf_{m\in B,w\in C} J_P(m,w)\right\}
\ee
and
\be
Q_n(M_n\in B, W_n\in C) \asymp \exp\left\{-n \inf_{m\in B,w\in C} J_Q(m,w)\right\}
\ee
for ``good'' sets $B\times C$. More concretely, we can also write
\be
P_n(M_n\in [m,m+dm], W_n\in [w,w+dw]) \asymp e^{-n J_P(m,w)}
\ee
and
\be
Q_n(M_n\in [m,m+dm], W_n\in [w,w+dw]) \asymp e^{-n J_Q(m,w)}.
\ee
As for Condition (c), it follows from our previous assumption that $P_n$ and $Q_n$ have the same support on $\Lambda_n$, so they also have the same support when pushed forward to $\cM\times\reals$ with the random variables $(M_n(\bX_n), W_n(\bX_n))$. This property is again not essential, but simplifies the derivation and analysis of our results.

We will show in Sec.~\ref{secexamples} how the two joint LDPs can be derived in practice using techniques from large deviation theory. The existence of these LDPs can be viewed as a strong assumption of our theory, but they are necessary, as will become clear, to fully understand the asymptotic efficiency of $Q_n$.

For now, we assume that two rate functions $J_P$ and $J_Q$ for $(M_n,W_n)$ are given and proceed to relate them to $I_P(B)$ and $R_Q(B)$. To this end, it is important to note that, although $J_P$ and $J_Q$ are defined with respect to $P_n$ and $Q_n$, respectively, both rate functions actually depend on $Q_n$, since they both involve the action $W_n$ defined from $L_n$. The joint rate functions are also linked, since expectations with respect to $Q_n$ are related to expectations with respect to $P_n$, and vice versa, via the identity
\be
\EX_Q[L_n(\bX_n)\, f(\bX_n)] = \EX_P[f(\bX_n)],
\label{eqrnd2}
\ee
where $f$ is any test function. This result was already used in \eqref{equnb1} to show that the IS estimator $\hp_n^N$ is unbiased, and implies the following large deviation result, referred to in physics as a \emph{fluctuation relation} \cite{harris2007}:

\begin{proposition}[{\cite[Prop.~2]{touchette2018}}] Under Assumption~\ref{hyp3}, the two rate functions $J_P$ and $J_Q$ are such that 
\label{prop1}
\be
J_P(m,w) = J_Q(m,w)+w
\label{eqfr1}
\ee
for all $(m,w)$ in their domain.
\end{proposition}

This result simply follows by applying \eqref{eqrnd2} with indicator functions to transform joint probability distributions as follows:
\begin{eqnarray}
P_n(M_n\in dm, W_n\in dw) &=&\EX_P[\idf_{M_n\in dm}\idf_{W_n\in dw}]\nonumber\\
&=&\EX_Q[e^{-nW_n}\idf_{M_n\in dm}\idf_{W_n\in dw}]\nonumber\\
&=& e^{-nw} Q_n(M_n\in dm,W_n\in dw).
\end{eqnarray}
Here, we have used $L_n=e^{-nW_n}$ and the shorthand $M_n\in dm$ to mean $M_n\in [m,m+dm]$ (similarly for $W_n\in dw$). Taking the large deviation limit then yields \eqref{eqfr1}. We refer to \cite{touchette2018} for a rigorous presentation of this argument, based on the definition of the LDP and Assumption~\ref{hyp3}.

The existence of a joint LDP for $(M_n,W_n)$ also implies that $M_n$ and $W_n$ satisfy the LDP individually. This ``marginalization'' of joint LDPs is covered in \cite[Prop.~1]{touchette2018} and follows in large deviation theory from the contraction principle \cite{ellis1985,dembo1998,hollander2000,touchette2009}, stated for convenience in Appendix~\ref{appcont}. The application of this principle to marginalize (viz., trace out) $W_n$, for example, gives the following representation of the rate function of $M_n$ with respect to $P_n$:
\be
I_P(m) = \inf_{w\in\reals} J_P(m,w).
\ee
Therefore,
\be
I_P(B) = \inf_{m\in B,w\in\reals} J_P(m,w) =\inf_{m\in B,w\in\reals} \{ w+J_Q(m,w)\},
\label{eqipb1}
\ee
where we have used Proposition~\ref{prop1} to obtain the second equality. Similar formulas apply with respect to $Q_n$, including
\be
I_Q(m)=\inf_{w\in\reals} J_Q(m,w),
\label{eqmargq1}
\ee
which is the rate function of $M_n$ associated with its LDP with respect to $Q_n$.

At this point, we formulate one more assumption needed to derive our main result:
\begin{assumption}
\label{hyp4}
There exists a unique, finite pair $(m^*,w^*)$ such that $J_Q(m^*,w^*)=0$.
\end{assumption}

This assumption means concretely that the pair $(M_n,W_n)$ satisfies the weak law of large numbers with respect to $Q_n$ (see \cite[Thm.~2.5]{ellis2000}), that is,
\be
\lim_{n\ra\infty} Q_n\left(M_n\in [m^*,m^*+dm],W_n\in [w^*,w^*+dw]\right)=1.
\ee
In this case, we say that $(m^*,w^*)$ is the \emph{typical value} or \emph{concentration point} of $(M_n,W_n)$ under $Q_n$. Since rate functions are positive, Assumption~\ref{hyp4} and $(\ref{eqmargq1})$ imply $I_Q(m^*)=0$, so that $m^*$ is also the typical value of $M_n$ with respect to $Q_n$. A good change of measure, as we have seen, should be such that $m^*\in \bB$ to transform the event $\{M_n\in B\}$ from being rare under $P_n$ to being typical under $Q_n$. In large deviation terms, this means 
\be
I_Q(B)\equiv \inf_{m\in B} I_Q(m)=I_Q(m^*)=0.
\ee
This is the first step for constructing a good change of measure for IS -- to make $B$ typical. The next step is to ensure that $Q_n$ is asymptotically efficient.

\subsection{Efficiency results}

We study the efficiency of $Q_n$ from the result in \eqref{eqeff1} by expressing $R_Q(B)$ as a variational formula involving $J_Q(m,w)$, similarly to the formula \eqref{eqipb1} that we have for $I_P(B)$, and by then comparing these two formulas to infer conditions on $J_Q(m,w)$ that guarantee that $R_Q(B)=2I_P(B)$. The first part is the subject of the next result. 

\begin{lemma} 
\label{lem1}
Under Assumptions~\ref{hyp1}, \ref{hyp2} and \ref{hyp3}, the second moment rate $R_Q(B)$ defined in \eqref{eqde1} exists, is finite, and is given in terms of $J_P$ and $J_Q$ by
\be
R_Q(B) = \inf_{m\in B,w\in\reals} \{w+J_P(m,w)\} = \inf_{m\in B,w\in\reals} \{2w+J_Q(m,w)\}. 
\label{eqmomres1}
\ee
\end{lemma}

\begin{proof}
These variational representations of $R_Q(B)$ are a direct consequence of the Laplace principle for approximating exponential integrals, which is formulated in a rigorous way in large deviation theory via Varadhan's integral lemma \cite{varadhan1966}. For our purpose, we apply a version of that theorem found in \cite[Thm.~II.7.2]{ellis1985} to $R_Q(B)$ as defined by \eqref{eqde1}. Given Assumption~\ref{hyp3}(d), to show that
\be
\lim_{n\ra\infty} -\frac{1}{n} \log \EX_Q[L_n^2 \idf_{M_n\in B}]=\inf_{m\in B,w\in\reals} \{2w+J_Q(m,w)\},
\ee
it suffices to prove that
\be
\liminf_{n\ra\infty} -\frac{1}{n} \log \EX_Q[L_n^2 \idf_{M_n\in \bar B}]\geq \inf_{m\in \bar B,w\in\reals} \{2w+J_Q(m,w)\},
\label{capzjd}
\ee
and
\be
\limsup_{n\ra\infty} -\frac{1}{n} \log \EX_Q[L_n^2 \idf_{M_n\in B^\circ}]\leq \inf_{m\in B^\circ,w\in\reals} \{2w+J_Q(m,w)\}.
\label{capzjdde}
\ee
Under Assumptions~\ref{hyp2} and \ref{hyp3}(b), to establish \eqref{capzjd} (respectively \eqref{capzjdde}), one may adapt the proof of \cite[Thm.~II.7.2]{ellis1985} detailed in Appendix B.2, item (a) (respectively (b)), replacing $K$ (respectively $G$) with $\bar B\times\reals$ (respectively $B^\circ\times\reals$). Hence, $R_Q(B)$ defined as the limit in \eqref{eqde1} exists. Moreover, Assumption \ref{apzicjicj} and \eqref{eqde1} ensure that $R_Q(B)<\infty$, and Assumption \ref{hyp2} that $R_Q(B)>-\infty$, so that $R_Q(B)$ is finite. Finally, \eqref{eqfr1} shows the equivalence between both relations for $R_Q(B)$ in \eqref{eqmomres1}.
\end{proof}

The result of Lemma~\ref{lem1} complements the methods developed by Bucklew \cite{bucklew2004} for calculating $R_Q(B)$, which are based on generating functions rather than the joint large deviations of $M_n$ and $W_n$.  The advantage of our result is that it can be used with \eqref{eqipb1} to express the efficiency bound $R_Q(B)\leq 2I_P(B)$ as a variational inequality involving the rate function $J_Q(m,w)$:
\be
\inf_{m\in B,w\in\reals} \{2w+J_Q(m,w)\}\leq 2 \inf_{m\in B,w\in\reals} \{w+J_Q(m,w)\}.
\label{eqeff2}
\ee
Therefore, \emph{$Q_n$ is asymptotically efficient if and only if $J_Q$ is such that the inequality above is an equality}. The same inequality can be expressed in terms of $J_P$ using \eqref{eqfr1}, but this is not useful, since we want to characterize the efficiency of $Q_n$. Note, however, that the right-hand side of \eqref{eqeff2}, although written with $J_Q$, does not actually depend on $Q_n$, since it is equal to $2I_P(B)$. 

Our aim now is to find conditions on $Q_n$, and therefore on $J_Q(m,w)$, to have equality in \eqref{eqeff2}. This is a non-trivial task, despite the simple form of this inequality, because the minimizers on either side need not be the same. Moreover, although $J_Q$ is positive, $w$ is not, so bounds based on the minimizer $(m^*,w^*)$ of $J_Q$ do not yield any useful conditions. Rather, such conditions are found by observing that $w$ is the relevant variable in \eqref{eqeff2}, since the minimizer over $m$ is common to both sides of this inequality, and that the unconstrained minimization over $w\in\reals$ has the form of a Legendre--Fenchel transform. 

Based on these observations, we define
\be
I_Q^B(w)\equiv \inf_{m\in \bar B} J_Q(m,w).
\label{eqdefiqb}
\ee
This function of $w\in\reals$ is positive, since $J_Q(m,w)\geq 0$, although it is not, as such, a rate function, since it does not necessarily have a zero. To understand this point, let us assume for simplicity that $B$ is closed. In that case, note that the joint LDP for $M_n$ and $W_n$ with respect to $Q_n$ implies the following LDP for the distribution of $W_n$ conditioned on $M_n\in B$:
\be
Q_n(W_n\in [w,w+dw]| M_n \in B) \asymp e^{-n I_Q(w|B)},
\ee
where
\be
I_Q(w|B) = \inf_{m\in B}J_Q(m,w) - \inf_{m\in B,w\in\reals} J_Q(m,w)=I_Q^B(w) - I_Q(B)
\label{eqdiffrf1}
\ee
The conditional distribution of $W_n$ is normalized, so its rate function $I_Q(w|B)$ is a true rate function, in the sense that
\be
\inf_{w\in\reals} I_Q(w|B) = 0.
\label{eqdiffrf2}
\ee
However, we see from \eqref{eqdiffrf1} that, unless $I_Q(B)=0$, we have $I_Q^B(w)>0$, so the latter function is indeed not a true rate function in general.

The case that interests us is precisely the case where $I_Q(B)=0$. That is, if $m^*\in \bar B$, then $B$ is typical under $Q_n$, so that $I_Q(B)=0$ and thus $I_Q(w|B)=I_Q^B(w)$. In that case, $I_Q^B(w)$ is interpreted as a conditional rate function having a zero (at $w^*$ from Assumption~\ref{hyp4}). The converse is also true, leading us to the following result:
\begin{lemma} 
\label{lem2}
$I_Q^B(w^*)\geq 0$ with equality if and only if $m^*\in \bB$.
\end{lemma}

\begin{proof} We have $I_Q^B(w^*)\geq 0$ by definition of rate functions. For the direct part, suppose that 
\be
0=I_Q^B(w^*)=\inf_{m\in \bar B} J_Q(m,w^*).
\ee 
As mentioned in the discussion before Assumption \ref{apzicjicj}, since $J_Q$ is a good rate function, the infimum is achieved on $\bar B$. By Assumption~\ref{hyp4}, this ensures that $m^*\in \bar B$.

For the converse part, simply note that $m^*\in \bar B$ implies 
\be
I_Q^B(w^*)=\inf_{m\in \bar B} J_Q(m,w^*)=J_Q(m^*,w^*)=0
\ee
by Assumption~\ref{hyp4}.
\end{proof}

We are now ready to state our main result for the asymptotic efficiency of $Q_n$ based on $I_Q^B$. The statement of the result uses the \emph{subdifferential} $\p I_Q^B(w^*)$ of $I_Q^B$ at the point $w^*$, which is the set of values $k\in\reals$ such that
\be
I_Q^B(w)\geq I_Q^B(w^*)+k(w-w^*)
\label{eqsubdiff1}
\ee
for all $w\in\reals$. More information about subdifferentials can be found in Appendix~\ref{appconvex}. For the proof and the interpretation of the result, we also need the \emph{Legendre--Fenchel transform} of $I_Q^B$, defined by\footnote{We use the same letter $\lambda$ for the Legendre--Fenchel transform and for the SCGF in \eqref{eqscgf1}, since, as already mentioned, the G\"artner--Ellis theorem ensures that, under appropriate conditions, both functions coincide.}
\be
\lambda_Q^B(k)\equiv \sup_{w\in\reals} \{kw-I_Q^B(w)\},\qquad k\in\reals.
\label{eqdefscgfiqb1}
\ee
This is a convex function of $k$, as also explained in Appendix~\ref{appconvex}, such that
\be
\lambda_Q^B(0)=-\inf_{w\in\reals} I_Q^B(w)\leq 0.
\label{eqlambdaineq1}
\ee

\begin{theorem} 
\label{thm1}
Under Assumptions~\ref{hyp1}-\ref{hyp4}, $Q_n$ is asymptotically efficient if and only if $I_Q^B(w^*)=0$ (typicality condition) and $-2\in\p I_Q^B(w^*)$ (steepness condition).
\end{theorem}

\begin{proof}
Suppose that $Q_n$ is asymptotically efficient. Then the efficiency criterion $R_Q(B)=2I_P(B)$ leads to equality in \eqref{eqeff2}, which can be re-expressed as
\be
\lambda_Q^B(-2) = 2\lambda_Q^B(-1)
\label{eqaffrel1}
\ee
with the definition of $I_Q^B$ and its Legendre--Fenchel transform. 

This relation constrains the graph of $\lambda_Q^B(k)$, as illustrated in Fig.~\ref{figlambda1}: $P_1$ and $P_2$ show the points of $\lambda_Q^B$ at $k=-1$ and $k=-2$, respectively, which are related in an affine way according to \eqref{eqaffrel1}. Moreover, we know that $\lambda_Q^B(0)\leq 0$ from \eqref{eqlambdaineq1}. From these two results, and the fact that $\lambda_Q^B(k)$ is convex, we conclude that $\lambda_Q^B(k)$ must be linear over $k\in [-2,0]$, as no other convex function can pass through both $P_1$ and $P_2$ while intersecting the ordinate axis below $0$. Hence,
\be
\lambda_Q^B(k) = -\lambda_Q^B(-1) k,\qquad k\in [-2,0].
\label{eqaffine1}
\ee 
In particular, $\lambda_Q^B(0)=0$, shown in Fig.~\ref{figlambda1}b as the point $P_0$ at the origin. By \eqref{eqlambdaineq1}, this implies that 
\be
0=\inf_{w\in\reals} I_Q^B(w)=\inf_{(m,w)\in\bar B\times\reals} J_Q(m,w).
\ee
Since $\bar B\times\reals$ is closed, the infimum is reached on the latter and, by Assumption~\ref{hyp4}, this is possible only at point $(m^*,w^*)$, so that $m^*\in \bB$ and $I_Q^B(w^*)=0$ by Lemma~\ref{lem2}.

The typicality condition is obtained from this result by noting, following the proof of Lemma~\ref{lem2}, that the existence of $P_0$ implies with \eqref{eqlambdaineq1} that $I_Q^B(w^*)=0$. The minimum $w^*$ is unique by Assumption~\ref{hyp4}. The steepness condition, on the other hand, is obtained by using standard results of convex analysis, stated with references in Appendix~\ref{appconvex}, to show that the linear part of $\lambda_Q^B$ leads to $I_Q^B$ having a cusp at its global minimum, characterized by more than one supporting lines, including one with slope $0$ and one with slope $-2$. To simplify the proof, we will first assume that $I_Q^B(w)$ is a convex function of $w$ and will then explain why the result also holds when $I_Q^B$ is not convex.

Note first that $w^*$ being a global minimum of $I_Q^B$ means that $0\in \p I_Q^B(w^*)$. Using the duality result expressed in \eqref{eqdual1}, we then obtain $w^*\in\p \lambda_Q^B(0)$. However, since $\lambda_Q^B(k)$ is linear for $k\in [-2,0]$, we also have $w^*\in\p \lambda_Q^B(k)$ for $k\in [-2,0]$, which implies by applying \eqref{eqdual1} again that
\be
[-2,0]\subset\p I_Q^B(w^*).
\label{eqsteeprel1}
\ee
Therefore, $-2\in I_Q^B(w^*)$, which is the steepness condition.

If $I_Q^B$ is nonconvex, then the same argument applies by replacing $I_Q^B$ in the duality result by its convex envelope $(I_Q^B)^{**}$, given by the Legendre--Fenchel transform of $\lambda_Q^B$ or, equivalently, by the double Legendre--Fenchel transform of $I_Q^B$ itself. We also have to note that a function and its convex envelope necessarily have the same global minima, if there are any, which means here that $I_Q^B(w^*)^{**}=I_Q^B(w^*)=0$ and $0\in\p I_Q^B(w^*)^{**}$. Finally, where a function coincides with its convex envelope, the subdifferentials are the same, so that $\p I_Q^B(w^*)^{**}=\p I_Q^B(w^*)$. All these results are presented with references in Appendix~\ref{appconvex} and imply, in the end, that the relation \eqref{eqsteeprel1} holds at $w^*$ even if $I_Q^B$ is not convex. 

\begin{figure}[t]
\centering
\includegraphics{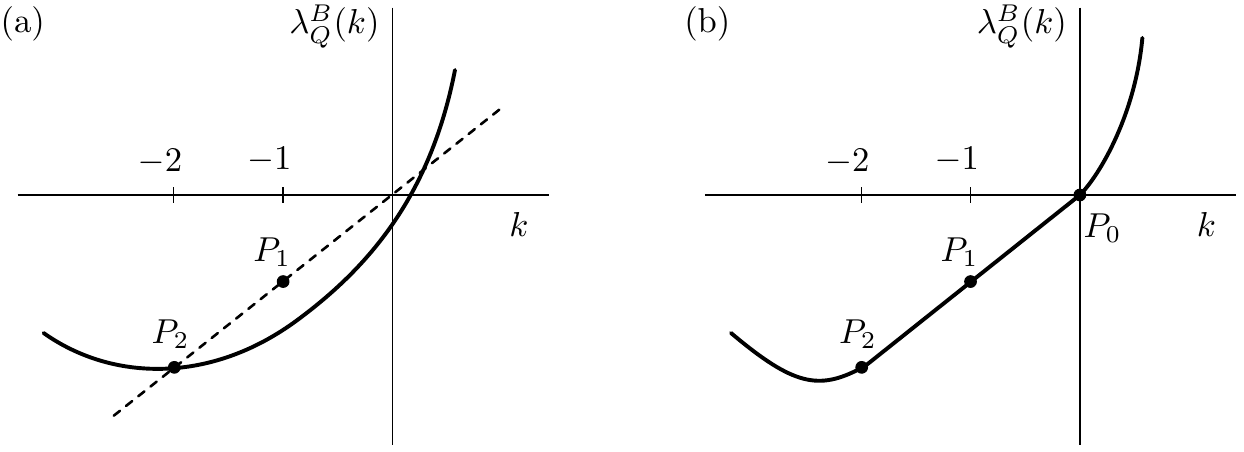}
\caption{(a) General $\lambda_Q^B(k)$. (b) Imposing $\lambda_Q^B(-2) = 2\lambda_Q^B(-1)$ implies, by convexity of $\lambda_Q^B(k)$, that this function passes through the origin $P_0$ and is linear for $k\in [-2,0]$.}
\label{figlambda1}
\end{figure}

To complete the proof, we consider the converse statement. We have again that, since $I_Q^B$ has a global minimum at $w^*$, $0\in \partial I_Q^B(w^*)$. By further assuming that $-2\in  \partial I_Q^B(w^*)$, we then have $[-2,0]\subset \partial I_Q^B(w^*)$, since subdifferentials are closed convex sets. Hence, $-1$ also belongs to the subdifferential of $I_Q^B(w^*)$, which means with \eqref{eqsubdiff1} that
\be
I_Q^B(w) \geq I_Q^B(w^*) -(w-w^*)
\ee
and, therefore, 
\be
\inf_{w\in\reals} \{w+I_Q^B(w)\}= w^*+I_Q^B(w^*) = w^*.
\label{eqmin1}
\ee
The same argument for $-2$ gives
\be
\inf_{w\in\reals} \{2w+I_Q^B(w)\}= 2w^*+I_Q^B(w^*) = 2w^*.
\label{eqmin2}
\ee
Consequently,
\be
\inf_{w\in\reals} \{2w+I_Q^B(w)\}=2\inf_{w\in\reals} \{w+I_Q^B(w)\}= 2 w^*,
\ee
which implies from \eqref{eqeff2} that $Q_n$ is asymptotically efficient.
\end{proof}

\subsection{Interpretation and special cases}

We will see in the next section that our main result in Theorem~\ref{thm1} covers the efficiency of the exponential tilting as a special case. The important contribution of this theorem, compared to previous results, is the subdifferential condition, which guarantees that the second moment
\be
F_n(B)\equiv \EX_Q[L_n^2 \idf_{M_n\in B}] =\int_{\cM\times\reals} e^{-2nw} \idf_{m\in B}\, Q_n(dm,dw).
\label{eq2m1}
\ee
entering in the definition of $R_Q(B)$ is dominated by $w^*$ and not by another rare value of the action smaller than $w^*$ that would lead to an exponentially larger value of the likelihood factor $L_n$. If this condition is satisfied, in addition to the obvious condition that $B$ be typical under $Q_n$, then $Q_n$ is asymptotically efficient, which means that it can be used to sample $P_n(M_n\in B)$ with a sample size $N_n$ according to \eqref{eqsamplesize1} that is not exponentially large with $n$.

Comparing \eqref{eqmin1} and \eqref{eqmin2}, we can also say that $Q_n$ is asymptotically efficient if and only if $w^*$ is the minimizer on both sides of the efficiency criterion
\be
\inf_{w\in\reals} \{2w +I_Q^B(w)\} = 2\inf_{w\in\reals} \{w+I_Q^B(w)\},
\ee
which follows from \eqref{eqeff2}. In other words, $Q_n$ is asymptotically efficient if and only if the IS estimator $\hp_n^N$ and its second moment are dominated by the same typical value $w^*$ of the action, yielding $I_P(B)=w^*$ and $R_Q(B)=2w^*$. 

This interpretation of the efficiency in terms of the typical value $w^*$ does not mean altogether that the likelihood $L_n$ or its action $W_n$ does not fluctuate. This is a very important point. The subdifferential condition is only a condition about the ``shape'' of $I_Q^B(w)$ \emph{below} the typical value $w^*$, which means that the fluctuations of $W_n$ \emph{above} $w^*$ are not constrained in any way. 

In many cases, we find that $I_Q^B(w)$ is a convex function and is left-differentiable at $w^*$. Then the subdifferential condition reduces to
\be
I_Q^B(w^{*-})'\leq - 2,
\ee
where ${I_Q^B}(w^{*-})'$ is the left-derivative of $I_Q^B$ at $w^*$. This result is illustrated in Fig.~\ref{figconvex1}a and explains why we refer to the subdifferential condition as a ``steepness'' condition. Obviously, if $I_Q^B(w)$ is convex and is differentiable at its minimum, then
\be
I_Q^B(w^{*-})' = I_Q^B(w^*)'=0
\label{eqzerotest1}
\ee 
and so $Q_n$ is not asymptotically efficient. This offers a simple test that can be used in practice to identify non-efficient IS measures.

\begin{figure}[t]
\centering
\includegraphics[width=\textwidth]{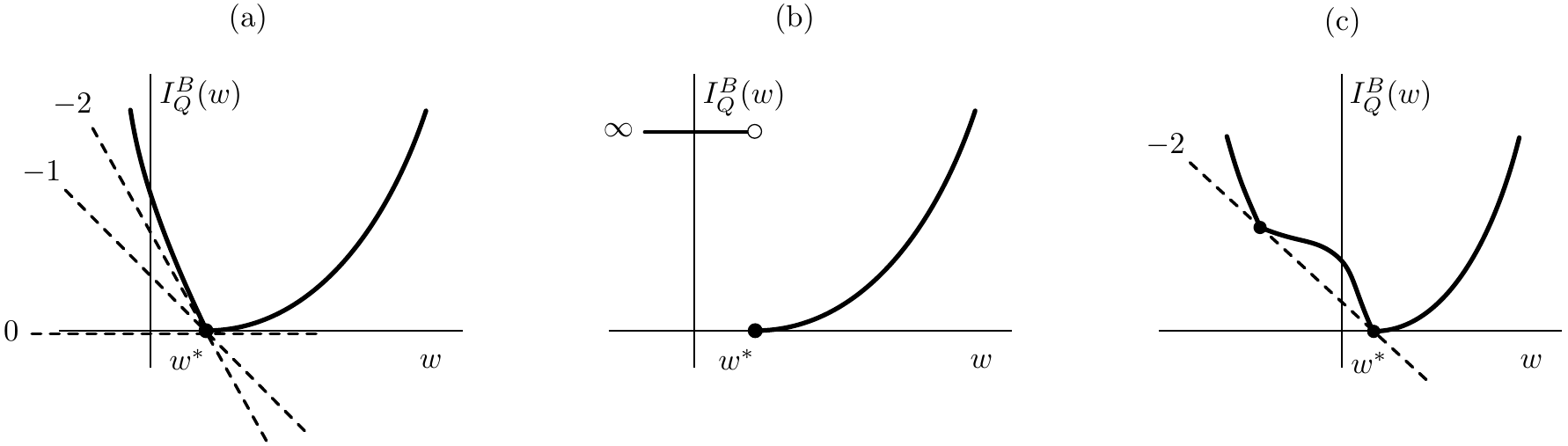}
\caption{(a) Efficiency condition for convex and left-differentiable $I_Q^B(w)$. (b) Asymptotically efficient $I_Q^B(w)$ diverging at the left of $w^*$. (c) Nonconvex $I_Q^B(w)$ that is also asymptotically efficient.}
\label{figconvex1}
\end{figure}

In general, $I_Q^B$ might not be left-differentiable at its minimum or be convex, contrary to $\lambda_Q^B$ which is convex by definition. This is why we need to state the steepness condition using the concept of subdifferentials. For instance, if $I_Q^B$ diverges for $w<w^*$, as illustrated in Fig.~\ref{figconvex1}b, then we have efficiency, since $ (-\infty, 0]\subset \p I_Q^B(w^*)$, so that $-2\in \p I_Q^B(w^*)$. This case arises when $W_n$ has no possible values below $w^*$ and so $W_n\geq w^*$. On the other hand, if $I_Q^B$ is not convex, then we have efficiency if $I_Q^B$ has a supporting line at $w^*$ with slope smaller than or equal to $-2$, as shown in Fig.~\ref{figconvex1}c. This follows from the interpretation of subdifferentials, explained in Appendix~\ref{appconvex}. Equivalently, we have efficiency if the left-derivative of the convex envelope of $I_Q^B$ at $w^*$ is smaller than or equal to $-2$. This case will be illustrated in the next section by revisiting the Gaussian sample mean and nonconvex set $B$ studied before.

With this geometric interpretation of efficiency, it is now clear that the likelihood factor $L_n$ does not have to be a deterministic function of $M_n$ or become so in the limit $n\ra\infty$ to efficiently estimate $P_n(M_n\in B)$, as often stated in studies of IS. The likelihood can fluctuate jointly with $M_n$ so long as the fluctuations of the action $W_n$ conditioned on $M_n\in B$ are sufficiently suppressed below the typical value $W_n=w^*$, that is, so long as $I_Q^B(w)$ is steep enough below $w^*$. The right steepness of $I_Q^B$ is not constrained in any way because values $W_n>w^*$ are exponentially suppressed in the integral \eqref{eq2m1}, which determines the efficiency of $Q_n$. Only values $W_n<w^*$ can affect the efficiency, as $L_n$ is then exponentially larger than the typical value $L_n^*=e^{-n w^*}$. In other words, accumulating likelihood factors that are exponentially \emph{smaller} than $L_n^*$ does not influence the efficiency of IS, but accumulating factors that are exponentially \emph{larger} than $L_n^*$ \emph{too frequently} does.

To close this section, we note that if there is a finite $w$ such that $-2-\delta \in \p I_Q^B(w)$ for some $\delta>0$, then under Assumption~\ref{hyp3}(b) the limit \eqref{eqass2} in our Assumption~\ref{hyp2} must be finite. As a result, we can rephrase that assumption in a more geometric and practical way as follows:

\medskip

\noindent\textbf{Assumption $\mathbf{2'}$.} \emph{$I_Q^B(w)$ must be coercive enough so that it has a point whose subdifferential contains a value strictly smaller than $-2$.}

\medskip

If $I_Q^B(w)$ is a convex and differentiable function, then this amounts to saying that there is a point whose slope is strictly smaller than $-2$.

\section{Examples}
\label{secexamples}

We illustrate in this section our results with simple examples to show how $I_Q^B(w)$ is calculated in practice and how its steepness determines the efficiency of IS estimators. We begin by considering the exponential tilting as a general change of measure, and then revisit the Gaussian and exponential \iid\ sample means, introduced in Sec.~\ref{secis}, as particular cases of that change of measure for which $I_Q^B(w)$ can be computed exactly. We also construct a variation of the exponential sample mean that shows that an IS estimator can be asymptotically efficient without having the exponential tilting form. This is an important result of this section.

We close the section with two examples related to Markov chains and stochastic differential equations to illustrate the generality of our formalism, to explain how the likelihood factor is defined for Markov processes, and to point out minor changes of notation when dealing with continuous-time processes. These examples should serve as a template to study the large deviations of more physical systems modelled by Markov processes in the context, for example, of nonequilibrium systems driven in steady-states and stochastic thermodynamics.

\subsection{Exponential tilting}
\label{secexamplesA}

We consider for simplicity the case where $M_n\in\reals$, since our goal is not to prove the efficiency of the exponential tilting in the most general setting but to illustrate our formalism based on $I_Q^B (w)$. The change of measure that we consider, as already defined in \eqref{eqexptilt1}, is thus
\be
Q_{n}(d\bX_n) = \frac{e^{n k M_n} P_n(d\bX_n)}{\EX_P[e^{n kM_n}]},
\ee
where $P_n(d\bX_n)$ is, as before, our prior measure or model of $\bX_n$ and $k$ is now a real parameter. 

The use of this exponential change of measure to study the large deviations of $M_n$ requires, as mentioned in Sec.~\ref{secis}, that the SCGF $\lambda_P(k)$ of $M_n$, as defined in \eqref{eqscgf1}, exists and is differentiable in $k$. In this case, it follows from the G\"artner-Ellis theorem \cite[Thm.~2.3.6]{dembo1998} that $I_P(m)$ is a strictly convex function, given by the Legendre--Fenchel transform of $\lambda_P(k)$, which means that it has a unique minimum and zero, denoted by $\barm$, corresponding to the typical value of $M_n$ under $P_n$. Thus, $I_P(\barm)=0$, which translates by Legendre duality into $\lambda_P'(0)=\barm$. 

In most applications, $I_P(m)$ is found to be a smooth (differentiable) function of $m$, so we will assume this property in this section to simplify the analysis. The particular form of the exponential change of measure then implies (see \cite[Thm.~2.4]{ellis2000}) that
\be
I_Q(m)=I_P(m)-km+\lambda_P(k),
\label{eqiqfct1}
\ee
so $I_Q(m)$ is smooth by assumption. It is also a good rate function, since $I_P(m)$ itself, as obtained from the G\"artner--Ellis theorem, is a good rate function.

With these results, we now apply our formalism by noting that the action $W_n$ of the exponential tilting is
\be
W_n=kM_n-c_n(k),
\label{eqexptiltf1}
\ee 
where
\be
c_n(k) = \frac{1}{n}\log \EX_P[e^{nkM_n}].
\ee
This is a deterministic function of $M_n$ that we write as $W_n=f_n(M_n)$, which implies that $J_Q(m,w)$ is defined only on the line $w=f_n(m)$ and is equal to $I_Q(m)$ on that line. The appearance of $n$ in this contraction appears a priori to be a problem; however, we show in Appendix~\ref{appcont} that, since $c_n(k)\ra\lambda_P(k)$ and $J_Q$ is a good rate function, $f_n$ can actually be replaced by the limit function $f(m) = km -c(k)$, where $c(k)=\lambda_P(k)$, consistently with \eqref{eqiqfct1} and \eqref{eqfr1}. As a result,
\be
I_Q^B(w) = 
\left\{
\begin{array}{lll}
\displaystyle\inf_{m\in \bar B}  I_Q(m) & & w=f(m)\\
\infty & & \text{otherwise.}
\end{array}
\right.
\ee

Having found $I_Q^B$, we now consider the rare event $M_n\in B\equiv[\hm,\infty)=\bar B$ with $\hm>\barm$. To make this event typical under $Q_n$, we fix $k$ so that the typical value $m^*$ of $M_n$ ``hits'' the boundary $\hm$. This is achieved by setting $k>0$ such that $\lambda_P'(k)=\hm$, leading to $I_Q(\hm)=0$ \cite{touchette2015} and
\be
\inf_w J_Q(\hm,w)=J_Q(\hm,f(\hm))=I_Q(\hm)=0.
\ee
This shows that we have a unique, typical pair $(m^*,w^*)=(\hm,f(\hm))$ under $Q_n$. Since $m^*=\hm\in \bar B$, we then have $I_Q^B(w^*)=0$ by Lemma~\ref{lem2}, so the first condition for efficiency in Theorem~\ref{thm1} is met.

To check the second condition, note that, since we have $k>0$ to achieve $m^*=\hm>\barm$, $J_Q(m,w)$ is not finite on $B$ when $w<w^*$, as shown in Fig.~\ref{figexptilt1}a, so that $I_Q^B(w)=\infty$ for all $w<w^*$. Thus, $I_Q^B$ is infinitely steep below $w^*$, which is sufficient, as mentioned before, to conclude that $Q_n$ is asymptotically efficient. Above $w^*$, we see instead that $J_Q(m,w)$ is finite on $B$, so $I_Q^B(w)$ is also finite for $w>w^*$. In fact, since $J_Q(m,w)$ has a unique zero at $(m^*,w^*)$, we have $0<I_Q^B(w)<\infty$ for $w>w^*$, showing overall that $I_Q^B(w)$ has the shape shown in Fig.~\ref{figconvex1}b, associated once again with a $Q_n$ that is asymptotically efficient.

\begin{figure}
\centering
\includegraphics{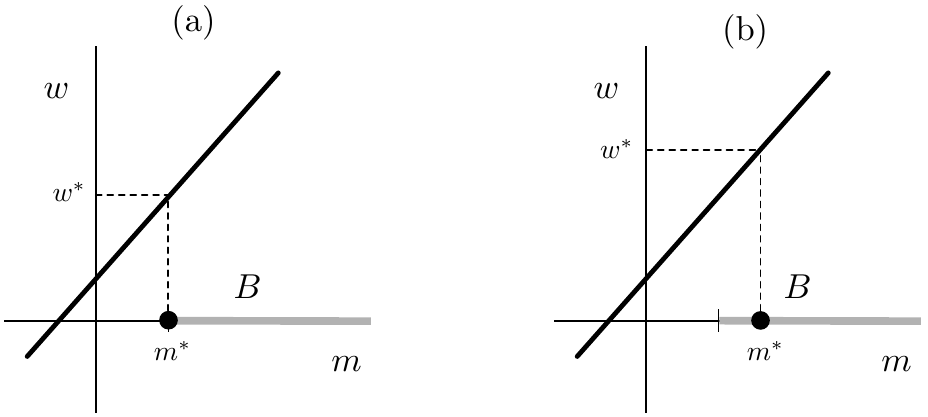}
\caption{Line $w=f(m)$ in the $(m,w)$ plane on which $J_Q(m,w)$ is defined for the exponential change of measure. (a) Asymptotically efficient $Q_n$ for which the typical value $m^*$ of $M_n$ is chosen on the boundary of $B$. (b) Non-efficient $Q_n$ associated with $m^*$ in the interior of $B$.}
\label{figexptilt1}
\end{figure}

The same argument can be used to show that $Q_n$ is not asymptotically efficient if $m^*$ is chosen inside $B$, that is, such that $m^*>\hm$. In this case, $I_Q^B(w)$ still has a zero at $w^*=f(m^*)$, but it does not diverge on the left of $w^*$ because the line $w=f(m)$ on which $J_Q(m,w)$ is finite does not ``go out'' of $B$ when $w$ goes below $w^*$; see Fig.~\ref{figexptilt1}b. Since $I_Q(m)$ is assumed to be smooth, $I_Q^B(w)$ must therefore have a smooth minimum in the vicinity of $w^*$ with zero derivative as in \eqref{eqzerotest1}, implying that $Q_n$ is not asymptotically efficient. 

Of course, the steepness condition could be satisfied in this case if $I_Q(m)$ had a steep-enough corner at $m^*$, but this would violate our assumption that $I_Q(m)$ is smooth, which is what is observed again in many applications.\footnote{A corner in $I_P(m)$ or $I_Q(m)$ signals physically a dynamical phase transition in the fluctuations of $M_n$. Here, we assume, for simplicity, that no such phase transition occurs. Note that a corner in the function $I_Q^B(w)$ is not related to a dynamical phase transition, since this function is obtained by conditioning. It can have a corner, as the example of the exponential tilting shows, regardless of whether $I_P(m)$ or $I_Q(m)$ is smooth.} With this assumption, the exponential tilting is therefore asymptotically efficient, as proved, if it ``hits'' $B$ on its boundary $\hm$ rather than in the interior of $B$. This can be generalized to $M_n\in\reals^D$ by requiring that $Q_n$ ``hit'' the dominating point of $B$, which is usually on the boundary of $B$; see \cite{sadowsky1990} for details.

All these results apply obviously if we change the rare event to $B=(-\infty,\hm]$ with $\hm<\barm$, in which case $k<0$. The efficiency of $Q_n$ is also direct if we consider the infinitesimal set $B=[\hm,\hm+dm]$. Then $k$ must be chosen such that $\lambda_P'(k)=\hm$ to achieve $m^*=\hm$, as mentioned before, which fixes $w^*=f(\hm)$ as the only action for which $I_Q^B(w)$ is finite. Thus, $I_Q^B(w)=0$ for $w=w^*$ and $\infty$ otherwise, which is obviously asymptotically efficient. This is a common case considered in physics, where the focus is usually on computing the rate function $I_P(m)$ rather than the probability $P_n(M_n\in B)$. In this case, one performs many simulations with different values of $k$ to estimate the probability of small, contiguous ``windows'' $[\hm,\hm+dm]$, which are converted with the large deviation limit \eqref{eqldp1} into points of $I_P(m)$ \cite{touchette2011}. 

Such a use of the exponential tilting in simulations is asymptotically efficient, as just shown, if $I_P(m)$ is a convex differentiable function and $B$ is a convex set. We have already seen in Sec.~\ref{secis} that the exponential tilting can be non-efficient if $B$ is nonconvex. By revisiting this example below, we will see that the problem in this case comes from the steepness condition controlling the asymptotic variance. On the other hand, the exponential tilting can also be non-efficient if $I_P(m)$ is nonconvex. The problem in this case is not the steepness of $I_Q^B$, and so the variance, but the fact that not all values of $M_n$ can be made typical by varying the tilting parameter $k$. This is related in physics to the nonequivalence of statistical ensembles; see \cite{touchette2018} for more details.

\subsection{Gaussian sample mean}

We now revisit the Gaussian sample mean studied in Sec.~\ref{secis} to show how the efficiency of $Q_n$ can be ascertained by calculating $I_Q^B$ explicitly using standard techniques from large deviation theory. This example also provides a further illustration of the exponential change of measure.

The setting is the same as in Example~\ref{exgauss1}: $P_n$ is the product measure $\cN(0,1)^{\otimes n}$ of $n$ \iid\ standard normal random variables, $M_n$ is their sample mean, and we look for the probability that $p_n=P_n(M_n\geq 1)$, so that $B=[1,\infty)$, leading to the rate exponent shown in \eqref{eqgaussex1}. 

We consider as the change of measure the product measure $Q_n=\cN(\mu,\sigma^2)^{\otimes n}$ associated with $n$ \iid\ Gaussian random variables with mean $\mu\in\reals$ and variance $\sigma^2>0$. The action for this change of measure, which is more general than the one considered in Example~\ref{exgauss1}, can be expressed as
\be
W_n =\frac{\mu}{\sigma^2}M_n +\left(\frac{\sigma^2-1}{2\sigma^2}\right)C_n-\frac{\mu^2}{2\sigma^2}-\log \sigma,
\label{eqactiongauss1}
\ee
where $M_n$ is the sample mean and 
\be
C_n=\frac{1}{n}\sum_{i=1}^n X_i^2
\ee
is the sample second moment. Since both $M_n$ and $C_n$ involve \iid\ random variables, we can use Cram\'er's theorem to find their joint rate function $K_Q(m,c)$ as the Legendre--Fenchel transform of the joint SCGF with respect to $Q_n$:
\be
\lambda_Q(k,\gamma) = \log \EX_Q[e^{kX+\gamma X^2}].
\ee
For $X\sim Q=\cN(\mu,\sigma^2)$, we find
\be
\lambda_Q(k,\gamma) =\frac{k^2\sigma^2/2+\mu  (\gamma  \mu +k)}{1-2 \gamma  \sigma ^2}-\frac{1}{2} \log \left(1-2\gamma  \sigma ^2\right)
\ee
for $1-2\gamma\sigma^2>0$. Accordingly,
\be
K_Q(m,c) = \sup_{k,\gamma}\{km+\gamma c-\lambda_Q(k,\gamma)\}= \frac{\sigma ^2 \log \left(\frac{\sigma ^2}{c-m^2}\right)+c+\mu ^2-2 \mu  m-\sigma ^2}{2 \sigma ^2},
\ee
which is defined for $m^2<c$ by the Cauchy-Schwarz inequality. Changing variables from $(M_n,C_n)$ to $(M_n,W_n)$, we then deduce
\be
J_Q(m,w) = K_Q(m,c(m,w)),
\label{eqctrans1}
\ee 
where
\be
c(m,w) =\frac{2 w\sigma^2-2m\mu+\mu^2+\sigma^2\log \sigma^2}{\sigma^2-1}.
\label{eqctrans2}
\ee
This holds if $\sigma\neq 1$. If $\sigma=1$, then $W_n$ is only a function of $M_n$,
\be
W_n =f(M_n)=\mu M_n -\frac{\mu^2}{2},
\ee
similarly to the exponential change of measure and, therefore,
\be
J_Q(m,w)=
\left\{
\begin{array}{lll}
(m-\mu)^2/2 & & w=f(m) \\
\infty & & \text{otherwise},
\end{array}
\right.
\label{eqgausssigma1}
\ee
given that $I_Q(m)=(m-\mu)^2/2$.

\begin{figure}
\includegraphics{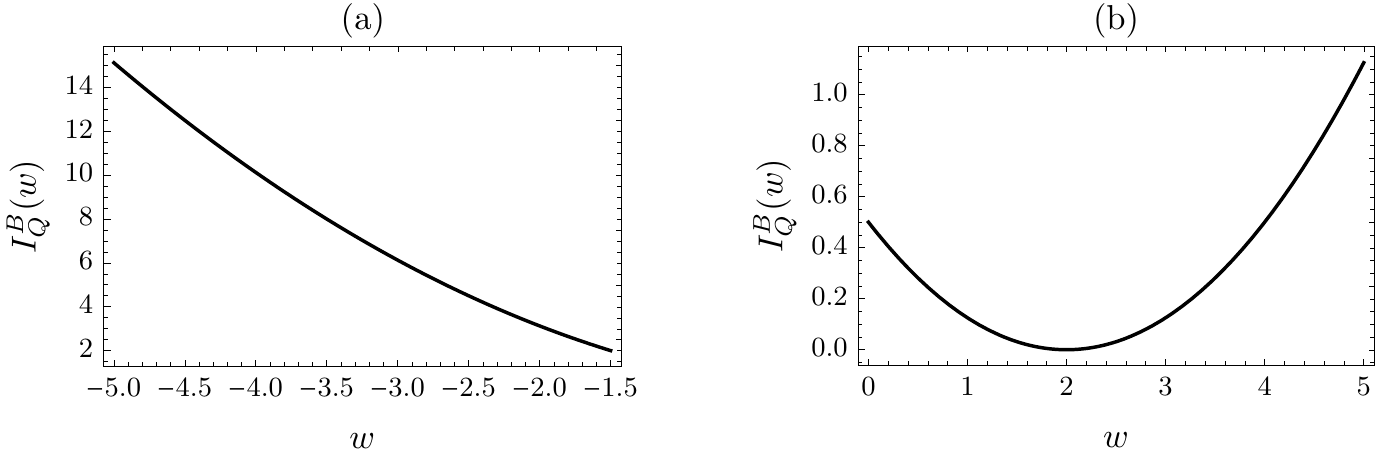}
\caption{$I_Q^B(w)$ for the Gaussian sample mean for (a) $\mu=-1$ and $\sigma=1$ and (b) $\mu=2$ and $\sigma=1$. Note that only the finite part of $I_Q^B(w)$ is shown.}
\label{figgausseff1}
\end{figure}

The efficiency of $Q_n$ is determined by computing $I_Q^B(w)$ from these explicit rate functions for various values of $\mu$ and $\sigma$. We begin with $\sigma=1$ and consider three cases for $\mu$, noting that $m^*=\mu$ and $w^*=f(m^*)=\mu^2/2$:
\begin{itemize}
\item $\mu<1$: $Q_n$ is not asymptotically efficient in this case simply because $m^*\notin \bar B$. This is confirmed by calculating $I_Q^B(w)$ from the contraction \eqref{eqdefiqb}. The result is shown for $\mu=-1$ in Fig.~\ref{figgausseff1}a: it does not have a zero, so $Q_n$ is indeed not asymptotically efficient.

\item $\mu\geq 1$: The calculation of $I_Q^B(w)$ gives in this case
\be
I_Q^B(w)=
\left\{
\begin{array}{lll}
(w/\mu-\mu/2)^2/2 & &w\geq\mu -\mu^2/2\\ 
\infty& &\text{otherwise}. 
\end{array}
\right.
\ee
For $\mu=1$, we have efficiency, since $I_Q^B(w)$ has its zero at $w^*=1/2$, implying $m^*\in \bar B$, and diverges left of $w^*$, so it is infinitely steep, as in Fig.~\ref{figconvex1}b. This is also confirmed by the fact that $Q_n$ is the exponential change of measure with $m^*=1$ at the boundary of $B$. For $\mu>1$, $I_Q^B(w)$ also has its zero at $w^*$ but $I_Q^B(w^{*-})'=0$, so it is not steep left of $w^*$, as shown in Fig.~\ref{figgausseff1}b. 
\end{itemize}
These results show overall that $Q_n$ is asymptotically efficient for $\sigma^2=1$ if and only if $\mu=1$.

For $\sigma^2\neq 1$, the contraction of $J_Q(m,w)$ leading to $I_Q^B(w)$ is more complicated to solve, since the minimization on $m\in B$ is further constrained by $m^2<c$ in the transformation \eqref{eqctrans1}. This results in a tedious constrained minimization problem, which can easily be solved numerically, however, to plot $I_Q^B(w)$ for any $\mu$ and $\sigma^2\neq 1$. Two representative solutions are shown in Fig.~\ref{figgaussnoneff1} and confirm our expectations from Theorem~\ref{thm1}. On the one hand, if $\mu<1$, then $Q_n$ is not asymptotically efficient since $m^*\notin \bar B$, as reflected by the fact that $I_Q^B(w)$ has no zero (Fig.~\ref{figgaussnoneff1}a). On the other hand, if $\mu\geq 1$, then $m^*\in \bar B$, so $I_Q^B(w)$ has a zero, but the rate function is not steep left of that zero, so $Q_n$ is still not asymptotically efficient (Fig.~\ref{figgaussnoneff1}b). This applies whether $\mu=1$ or $\mu>1$, which means in the end that $Q_n$ is not asymptotically efficient when $\sigma^2\neq 1$.

\begin{figure}
\includegraphics{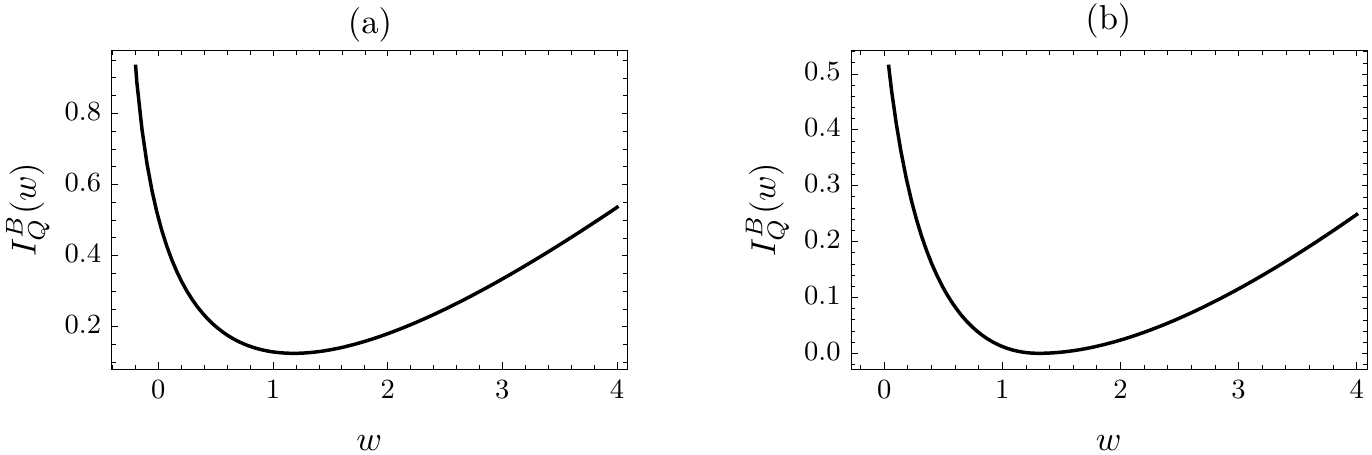}
\caption{$I_Q^B(w)$ for the Gaussian sample mean for (a) $\mu=0$ and $\sigma=2$ and (b) $\mu=1$ and $\sigma=2$. Note that only the finite part of $I_Q^B(w)$ is shown.}
\label{figgaussnoneff1}
\end{figure}

\subsection{Nonconvex $B$}

We can use the results of the previous section to understand the efficiency of $Q_n$ in Example~\ref{exnonconv1}, presented earlier to show that the exponential tilting can be non-efficient when $B$ is nonconvex. The setting is the same as in the previous section, except that $B$ is now chosen to be $B=(-\infty,-b]\cup [1,\infty)$ with $b>1$. We also consider $\mu=1$ and $\sigma^2=1$, which leads to
\be
I_Q^B(w)=
\left
\{\begin{array}{lll}
(w-1/2)^2/2&&w\geq 1/2 \text{ or } w\leq -b-1/2\\ 
\infty&&\text{otherwise}. 
\end{array}
\right.
\ee
This function is plotted in Fig.~\ref{figgauss1}. It has one zero at $w^*=1/2$, confirming that $m^*=1\in \bar B$ and a supporting line joining the two extremal points of $I_Q^B(w)$ at $w=-b-1/2$ and $w=w^*=1/2$, as shown in the figure, whose slope is $-(b+1)/2$. This is the supporting line with smallest slope, so $-2\in \p I_Q^B(w^*)$ if and only if $b\geq 3$, confirming the result of \cite{bucklew2004} announced in Example~\ref{exnonconv1}.

We can generalize this result by calculating $I_Q^B(w)$ for $\mu\neq 1$ to conclude that there is no other efficient parameters and, thus, that $Q_n$ is actually efficient if and only if $\mu=1$ and $b\geq 3$. This follows by considering three cases:
\begin{itemize}
\item $\mu\leq -b$. Then
\be
I_Q^B(w)=
\left
\{
\begin{array}{lll}
(w/\mu-\mu/2)^2/2& &w\leq\mu -\mu^2/2\ \mbox{or}\ w\geq -\mu b-\mu^2/2\\ 
\infty&&\text{otherwise}  
\end{array}
\right.
\ee
From this result, it can be checked that $I_Q^B(w^{*-})'=0$ if $\mu< -b$, so $Q_n$ is not asymptotically efficient. Then, if $\mu= -b$, one has $-2\in \p I_Q^B(w^*)$ if and only if $0<b\leq1/3$, which contradicts our assumption that $b>1$.

\item $-b<\mu<1$: Then $I_Q^B(w)$ does not have a zero, as expected from the fact that $m^*\notin \bar B$, so $Q_n$ is not asymptotically efficient.

\item $\mu>1$: Then
\be
I_Q^B(w)=
\left\{
\begin{array}{lll}
(w/\mu-\mu/2)^2/2&& w\geq\mu -\mu^2/2 \text{ or } w\leq -\mu b-\mu^2/2\\ 
\infty&&\text{otherwise}  
\end{array}
\right.
\ee
In this case $I_Q^B(w^*)=0$ at $w^*=\mu^2/2$, but $I_Q^B(w^{*-})'=0$.
\end{itemize}

\begin{figure}[t]
\centering
\includegraphics{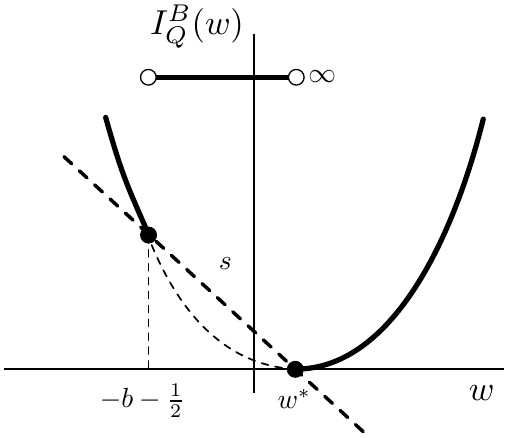}
\caption{Rate function $I_Q^B(w)$ for the Gaussian sample mean and nonconvex set $B$.}
\label{figgauss1}
\end{figure}

\subsection{Partial exponential tilting}

It is clear from the form of the exponential measure \eqref{eqexptilt1} that an \iid\ measure remains \iid\ when tilted by an additive functional $M_n$, which explains why such a change of measure is almost always considered when dealing with \iid\ sample means. In principle, other changes of measure that are not independent, not identically distributed or both could be considered and proved efficient within the formalism developed here. 

As a simple example, it can be checked for the Gaussian sample mean that changing all but one of the random variables is still asymptotically efficient, even though the resulting $Q_n$ is only a \emph{partial} exponential tilting (only $n-1$ random variables are tilted). This arises because the action of the partial exponential tilting differs from the action of the full exponential tilting by a sub-extensive term in $n$ that does not influence the large deviations of the action at the scale (or speed) $n$.

Surprisingly, this argument cannot be generalized to all \iid\ sample means and, in particular, not to the sample mean of exponential random variables considered in Example~\ref{exexpsm1}. In this case, we have seen that the product measure $Q_n=\mathcal{E}(\theta)^{\otimes n}$ of exponentials with parameter $\theta$, which changes the mean of all the random variables from $1$ to $1/\theta$, is asymptotically efficient if $\theta=1/b$. This can be checked by calculating $I_Q^B(w)$ explicitly, following the calculations of the previous sections or from the fact that $Q_n$ is the exponential tilting. 

The surprising result comes if we change the first $n-1$ random variables from $\cE(1)$ to $\cE(\theta)$, but keep the last one as $X_n\sim \cE(1)$. The action then is
\be
W_n = (1-\theta) \frac{n-1}{n} M_{n-1}+\frac{n-1}{n} \log \theta =(1-\theta)c_n M_{n-1} +c_n \log\theta,
\label{eqcontfct1}
\ee
where $M_{n-1}$ is the sample mean of the first $(n-1)$ random variables and $c_n=(n-1)/n$. Similarly, we can write
\be
M_n =\frac{n-1}{n} M_{n-1} +\frac{X_n}{n}=c_n M_{n-1}+T_n,
\label{eqcontfct2}
\ee
defining the new random variable $T_n=X_n/n$, which is independent of $M_{n-1}$. 

From these expressions, we find the joint rate function $J_Q(m,w)$ of $M_n$ and $W_n$ by noting that $M_{n-1}\sim\Gamma(n-1,(n-1)\theta)$, so this random variable satisfies the LDP with rate function
\be
I_\Gamma(y) = \theta  y-1-\log(\theta  y),
\ee
for $y\geq 0$, whereas $T_n\sim\cE(n)$, so it satisfies the LDP with rate function
\be
I_\cE(t)= t
\ee
also for $t\geq 0$. Both are good rate functions. From \eqref{eqcontfct1}, \eqref{eqcontfct2} and the fact that $c_n\ra 1$, we can then use the contraction principle presented in Appendix~\ref{appcont} to express $J_Q(m,w)$ as 
\be
J_Q(m,w) = \inf_{\substack{w=(1-\theta  )y +\log \theta  \\ m=y+t\\ y,t\geq 0}} I_\Gamma(y) +I_\mathcal{E}(t).
\label{eqcontpartialexp1}
\ee
In the latter, we have $m\leftrightarrow M_n\geq 0$, $w\leftrightarrow W_n$, $y\leftrightarrow M_{n-1}\geq 0$, and $t\leftrightarrow T_n\geq 0$. The solutions to the constraints are
\be
y = \frac{w-\log \theta  }{1-\theta  }\geq 0\hspace{1cm}\mbox{and}\hspace{1cm}
t = m-\frac{w-\log \theta  }{1-\theta  }\geq 0,
\ee
leading to $J_Q(m,w)=\infty$ if one of these constraints is not satisfied and
\be
J_Q(m,w)=m-w-\log\frac{w-\log \theta  }{1-\theta  }-1
\ee
otherwise. This rate function is good and has a single zero at $m^*=1/\theta$ and
\be
w^*=\frac{1-\theta}{\theta}+\log \theta.
\ee

At this point, we determine the asymptotic efficiency of $Q_n$, as before, by computing $I_Q^B(w)$ for different cases for $\theta$:
\begin{itemize}
\item $\theta> 1/b$: In this case, we do not even need to calculate $I_Q^B(w)$: $m^*=1/\theta\leq 1<b$, so that $m^*\notin \bar B$, implying that $Q_n$ is not asymptotically efficient.

\item $0<\theta <1/b$: A direct calculation based on the fact that the map $m\mapsto J_Q(m,w)$ is increasing gives in this case
\be
I_Q^B(w)=\frac{w-\log \theta  }{1-\theta  }-w-\log\frac{w-\log \theta  }{1-\theta  }-1
\ee
for all $w\in[(1-\theta  )b +\log \theta,w^*]$. From this result, we find $I_Q^B(w^{*-})'=0$, so that $Q_n$ is once again not asymptotically efficient.

\item $\theta=1/b$: A similar calculation as before now yields
\be
I_Q^B(w)=b-w-\log\frac{w-\log \theta  }{1-\theta  }-1=b-w-\log\frac{w+\log b  }{1-1/b  }-1,
\ee
for all $w\in(\log \theta,w^*]$. As a result, 
\be
I_Q^B(w^{*-})'=-1-\frac{1}{b-1},
\ee
which is smaller than $-2$ if and only if $b\leq 2$.
\end{itemize}

The conclusion is that $Q_n$ is asymptotically efficiency if and only if $\theta=1/b$ and $b\leq 2$, so the partial exponential tilting does not have the same efficiency as the full exponential tilting.

This result is special to the exponential distribution because $X_n/n$ in this case has large deviations at the same scale as $M_n$ and $W_n$, and so affects both random variables in the contraction \eqref{eqcontpartialexp1}. By contrast, if we choose $X_n\sim\cN(\mu,\sigma^2)$, then it can be checked that $X_n/n$ satisfies the LDP at the scale $n^2$ so adding or removing a  Gaussian random variable from a sample mean has no effect on its large deviations. The same applies to sample means of bounded random variables and, more generally, random variables whose distribution decays faster than exponentially.

\subsection{Markov chains}

We move away from \iid\ models to consider discrete-time Markov chains evolving on a set $\Omega$. We assume, for simplicity, that $\Omega$ is finite and that the transition kernel $p(x,y) = P(X_{i+1}=y|X_i=x)$ is homogeneous and defines an ergodic Markov chain. Starting with an initial distribution $\rho(x) = P(X_1=x)$, the probability model is then expressed as
\be
P_n(X_1,\ldots, X_n)=\rho(X_1)p(X_1,X_2)\cdots p(X_{n-1},X_n)
\ee 
for all $\bX_n=(X_1,\ldots,X_n) \in \Omega^n$. 

The observable $M_n$ is still a function of the configuration $\bX_n$, now interpreted as a trajectory in discrete time, from which we define the rare event probability $p_n=P_n(M_n\in B)$. We assume as before that $M_n$ satisfies the LDP with respect to $P_n$ with good rate function $I_P$ and consider a change of model $Q_n$ to  sample $p_n$ with the IS estimator \eqref{eqisest1}. The choice of $Q_n$ depends, as always, on the observable considered. For additive functionals having the general form
\be
M_n =\frac{1}{n}\sum_{i=1}^{n-1} g(X_i,X_{i+1}),
\label{eqaddfct1}
\ee
$Q_n$ is usually chosen to be another ergodic Markov chain with transition kernel $q(x,y)$, absolutely continuous with respect to $p(x,y)$, so that
\be
Q_n(X_1,\ldots,X_n) = \rho(X_1)q(X_1,X_2)\cdots q(X_{n-1},X_n),
\ee
using the same initial distribution. In this case, the action simply is
\be
W_n = \frac{1}{n}\sum_{i=1}^{n-1} \log\frac{q(X_i,X_{i+1})}{p(X_i,X_{i+1})},
\ee
so both $M_n$ and $W_n$ are additive functionals of the Markov chain.

With this property, the joint large deviations of $M_n$ and $W_n$ can be obtained by standard techniques from large deviation theory (see, e.g., \cite[Sec.~3.1]{dembo1998}). Define the joint SCGF of $M_n$ and $W_n$ with respect to $Q_n$ as
\be
\lambda_Q(k,\gamma)=\lim_{n\ra\infty} \frac{1}{n}\log \EX_Q[e^{nkM_n +n\gamma W_n}].
\ee
It is known that this function is given by the logarithm of the principal eigenvalue $\zeta_Q(k,\gamma)$ of the so-called \emph{tilted transition matrix}, defined by
\be
q_{k,\gamma}(x,y) =q(x,y) e^{k g(x,y)+\gamma h(x,y)},
\ee
where $h(x,y) = \log (q(x,y)/p(x,y))$. Thus,
\be
\lambda_Q(k,\gamma) = \log \zeta_Q(k,\gamma).
\ee
From this result, the rate $J_Q(m,w)$ is then found from the G\"artner--Ellis theorem by taking the Legendre--Fenchel transform of $\lambda_Q(k,\gamma)$, similarly to the Gaussian sample mean example. From there, we find $I_Q^B(w)$, as before, by minimising $J_Q(m,w)$ on $m\in B$ and use, finally, this function to determine the efficiency of $Q_n$. These steps can be implemented analytically for small Markov chains with a few states, while larger chains can be dealt with numerically using standard eigenvalue packages. 

In most applications, $Q_n$ is chosen to be the exponential tilting, which for a Markov chain and additive $M_n$ is known to be another Markov chain with transition kernel
\be
q(x,y) =\frac{e^{k g(x,y)}r_{k}(y)}{r_{k}(x)\zeta_P(k)} p(x,y),
\label{eqmcexptilt1}
\ee
where $\zeta_P(k)$ is the dominant eigenvalue of the tilted matrix
\be
p_{k}(x,y) =p(x,y) e^{k g(x,y)},
\label{eqtiltmat1}
\ee
and $r_{k}$ is the associated (right) eigenvector. In this case, it easy to verify that $W_n$ is given by \eqref{eqexptiltf1} modulo boundary terms involving $r_k(X_1)$ and $r_k(X_n)$, which can be neglected as they do not play a role in the large deviations of $W_n$ when $\Omega$ is finite. 

The Markov kernel \eqref{eqmcexptilt1} has been discussed in many contexts, including queuing theory \cite{asmussen2007}, simulations \cite{bucklew2004},  and statistical physics \cite{chetrite2014}, and can be seen as a generalization of Doob's $h$-transform arising in ``bridge-like'' conditionings of Brownian motion and other Markov processes; see \cite[Sec.~4.2]{chetrite2014} for details.

As a simple application, let us consider the symmetric binary Markov chain with $X_i\in \{0,1\}$ and transition matrix
\be
p=
\left(
\begin{array}{cc}
1-\alpha & \alpha\\
\alpha & 1-\alpha
\end{array}
\right),
\ee 
where $\alpha \in (0,1)$. Considering the observable to be the sample mean
\be
M_n = \frac{1}{n}\sum_{i=1}^n X_i,
\ee
which gives the fraction of $1$'s in $\bX_n$, we can formulate two different changes of process that are asymptotically efficient. The first is the exponential tilting in \eqref{eqmcexptilt1}, for which $\zeta_P(k)$ and $r_k$ can be computed analytically as the principal eigenvalue and eigenvector of
\be
p_{k}=
\left(
\begin{array}{cc}
1-\alpha & \alpha\\
\alpha e^{k} & (1-\alpha)e^k
\end{array}
\right),
\ee
obtained from \eqref{eqtiltmat1} using $g(x,y)=x$. The asymptotic efficiency of the resulting Markov chain is determined by the previous results on the exponential change of measure (see Section \ref{secexamplesA}), and follows again from the fact that $W_n$ is a function of $M_n$. Details can be found in \cite[Thm.~3]{bucklew1990b}.

Surprisingly, the exponential tilting is not the only modified Markov chain for which $W_n$ is a function of $M_n$. We can also take the transpose of the $2\times 2$ matrix $p_k$ above, which has the same principal eigenvalue as $p_k$, and normalize the rows to obtain the transition matrix
\be
q=
\left(
\begin{array}{cc}
\frac{1-\alpha}{F_0} & \frac{\alpha e^k}{F_0}\\
\frac{\alpha}{F_1} & \frac{(1-\alpha)e^k}{F_1}
\end{array}
\right),
\label{eqnotexptilt1}
\ee
where $F_0=1-\alpha+\alpha e^k$ and $F_1=\alpha+(1-\alpha)e^k$. It can be checked that the action induced by this transition matrix, which is obviously different from \eqref{eqmcexptilt1}, is
\be
W_n= k M_n -(1-M_n) \log F_0- M_n \log F_1,
\ee
modulo unimportant boundary terms, so that $W_n$ is an affine function of $M_n$. Consequently, we have found another example of process that is potentially asymptotically efficient and yet is not the exponential tilting. The difference is that the value of $k$ in \eqref{eqnotexptilt1} fixing the typical value $M_n=b$ under $Q_n$ is not specified by the relation $\lambda_P'(k)=b$, which is special to the exponential tilting. This is not important for simulations, as we only need in practice a parameter that can be varied to fix any typical value of $M_n$, whatever the relation between the two.

In principle, other efficient changes of process could be constructed using, for example, higher-order Markov chains or even non-Markovian processes whose measure $Q_n(X_1,\ldots, X_n)$ does not factorize as a product of transition probabilities. Very little, unfortunately, is known about non-Markovian processes and their large deviations \cite{harris2009}. The main reason for considering the exponential tilting is that it is known to be a Markov chain when the underlying measure $P_n$ is Markovian and the observable $M_n$ is additive in time \cite{kuchler1998,chetrite2014}. If one considers, for instance, the square of a sample mean as the observable $M_n$, then the exponential tilting is not Markovian.

\subsection{Diffusion processes}

The application of our results to Markov processes evolving in continuous time follows the examples above with minor changes of notations and techniques developed in large deviation theory to deal with this class of processes. For this reason, we do not cover this class in details, but only indicate the main changes involved, focusing as a specific example on diffusion processes $(X_t)_{t\geq 0}$ in $\reals$, described by the following stochastic differential equation (SDE):
\be
dX_t = F(X_t) dt +\sigma(X_t) dB_t.
\label{eqsde1}
\ee
Here, $B_t$ is a Brownian motion in $\reals$, while $F$ and $\sigma$ are two real functions of $X_t$, known as the drift and the noise amplitude, respectively. Many different observables can be defined in the context of SDEs, depending on the application and large deviation limit (low-noise or long-time) considered. We can consider, for example,
\be
M_T=\frac{1}{T}\int_0^T f(X_t) dt
\label{eqobssde1}
\ee
as a generalisation of the sample means studied before, which leads us to the problem of estimating the probability $P_T(M_T\in B)$ in the limit $T\ra\infty$, where $P_T$ is the probability measure of the process $X_t$ over the time interval $[0,T]$ induced by the SDE \eqref{eqsde1}.
 
Contrary to discrete-time Markov chains, we cannot write down any explicit expression for $P_T$; however, there is an explicit expression for the Radon--Nikodym derivative associated with a change of process if we consider that process to result from a change of drift. That is to say, change the drift $F$ in \eqref{eqsde1} to obtain a new SDE 
\be
d X_t = G(X_t)dt+\sigma (X_t)dB_t,
\ee
which defines a new law for $(X_t)_{t=0}^T$ that we denote by $Q_T$. Then the action of this process, as compared to the original one, is obtained from Girsanov's theorem \cite[Sec.~6.4]{stroock1979}, which states that
\be
L_T =\frac{dP_T}{dQ_T} 
=
\exp
\left(
\int_0^T c(X_t) dB_t-\frac{1}{2}c(X_t)^2 dt
\right)
\ee
where 
\be
c(x) = \frac{F(x)-G(x)}{\sigma(x)}.
\ee
Consequently, 
\be
W_T = -\frac{1}{T}\log \frac{dP_T}{dQ_T} = \frac{1}{2T}\int_0^T c(X_t)^2 dt-\frac{1}{T}\int_0^T c(X_t) dB_t.
\label{eqdiffaction1}
\ee
Both $M_T$ and $W_T$ are functions of the trajectory $(X_t)_{t=0}^T$ with law $Q_T$ over $[0,T]$. 

From this result, the joint large deviations of $M_T$ and $W_T$ with respect to $Q_T$ can be obtained, similarly to Markov chains, by solving a spectral problem in which the transition matrix is replaced by the infinitesimal generator of $X_t$ \cite{touchette2017}. As for Markov chains, the notion of exponential change of measure can also be defined for continuous-time processes and involves spectral elements related to the large deviations of $M_T$ with respect to $P_T$. This is fully explained in \cite{chetrite2014}.

As a simple illustration of the exponential change of measure, consider the Ornstein--Uhlenbeck process in $\reals$, defined by
\be
dX_t = -\gamma X_t dt+\sigma dB_t,
\ee
where $\gamma>0$ and $\sigma>0$. Moreover, let us take 
\be
M_T = \frac{1}{T}\int_0^T X_t dt
\ee
as the observable, which represents the area of $X_t$ per unit time. In this case, it can be shown (see \cite[Sec.~6]{chetrite2014} for the full calculation) that the exponential change of measure associated with $X_t$, corresponding to the process version of \eqref{eqexptilt1}, is another SDE with drift $G(x)=-\gamma(x-m)$ and noise amplitude $\sigma$. For this new process, the typical value of $M_T$ is clearly $m$, so the process is asymptotically efficient for estimating the large deviation probability of $M_T\in B$ with $B=[m,\infty)$, $B=(-\infty, m]$ or $B=[m,m+dm]$. In all cases, we  find from \eqref{eqdiffaction1} that the typical value of $W_T$ under $Q_T$ is
\be
w^*=\frac{\gamma^2 m^2}{2\sigma^2},
\ee
which is the known rate function $I_P(m)$ of $M_T$ with respect to $P_T$. 

This result is a diffusion analog of the Gaussian sample mean studied before, for which we found that the exponential tilting is another Gaussian with translated mean. Here, we see that a Gaussian process tilted with the sample mean is a Gaussian process having the same variance but a different mean. It can be checked that, as for the \iid\ Gaussian sample mean, this is the only efficient change of measure in the class of Gaussian processes with linear drift. If we change the friction coefficient $\gamma$ to another value, in addition to adding a constant to change the mean, then the process is no longer asymptotically efficient for the same reason that changing the variance in the Gaussian sample mean is not efficient. The calculations for the SDE are more complicated, but the results are similar.

Applications of IS for diffusions have been studied in statistical physics \cite{kundu2011} as well as more applied areas such as finance and queueing theory, focusing invariably on the exponential change of measure \cite{asmussen2007}. In future works, it would be interesting to apply our formalism to study other IS measures for Markov processes, such as the one proposed in \cite{klymko2018,whitelam2018c,jacobson2019}, to determine their efficiency and to see, in the end, if there is any gain from not using the exponential tilting, which is difficult to construct in practice, since it involves the solution of a spectral problem whose knowledge is equivalent to solving the large deviation problem \cite{chetrite2015}. Another important problem is to determine whether our formalism can be applied to study the efficiency of IS in the low-noise limit of SDEs, which is extensively used in physics, chemistry and engineering to study rare transition pathways \cite{cottrell1983,freidlin1984,graham1989,luchinsky1998}. The IS method itself can be applied in this limit (see, e.g., \cite{eijnden2012}), but it is not clear to what extent our assumptions hold.

\appendix
\section{Convex analysis}
\label{appconvex}

We collect in this section basic results of convex analysis used in the paper in relation to the rate function $I_Q^B(w)$, defined in \eqref{eqdefiqb}, and its Legendre--Fenchel transform $\lambda_Q^B(k)$, defined in \eqref{eqdefscgfiqb1}. Both are functions of a single real variable, so we state the necessary results only for this simple case. We assume further that all convex functions are proper closed convex functions. For more general results and proofs, we refer to \cite{rockafellar1970,rockafellar1988,borwein2006}.

\subsection{Subdifferentials}

Let $f:\reals\ra\treals$ be a real function taking values in the set of extended reals $\treals$. The \emph{subdifferential} $\p f(x)$ of $f$ at the point $x$ is the set of all values $k\in\reals$ such that
\be
f(y)\geq f(x) +k(y-x)
\ee
for all $y\in\reals$ \cite[Sec.~23]{rockafellar1970}. Put differently, and as illustrated in Fig.~\ref{figconv1}a, $\p f(x)$ is the set of slopes of all possible supporting lines of $f$ at $x$. If $f$ has not supporting line at $x$, then $\p f(x)=\emptyset$. We will see next that this may happen when $f$ is nonconvex.

For convex functions, subdifferentials exist everywhere in the domain of $f(x)$, except possibly at boundary points \cite[Thm.~23.4]{rockafellar1970}. For this class of functions, we have in fact $\p f(x) = [f'(x^-),f'(x^+)]$, where $f'(x^-)$ is the left-derivative and $f'(x^+)$ the right-derivative \cite[Thm.~24.3]{rockafellar1970}. If these are equal, $f$ is differentiable at $x$ so that $\p f(x) = \{f'(x)\}$ \cite[Thm.~25.1]{rockafellar1970}. In all cases, $\p f(x)$ is a closed convex interval \cite[p.~215]{rockafellar1970}.

\begin{figure}[t]
\centering
\includegraphics{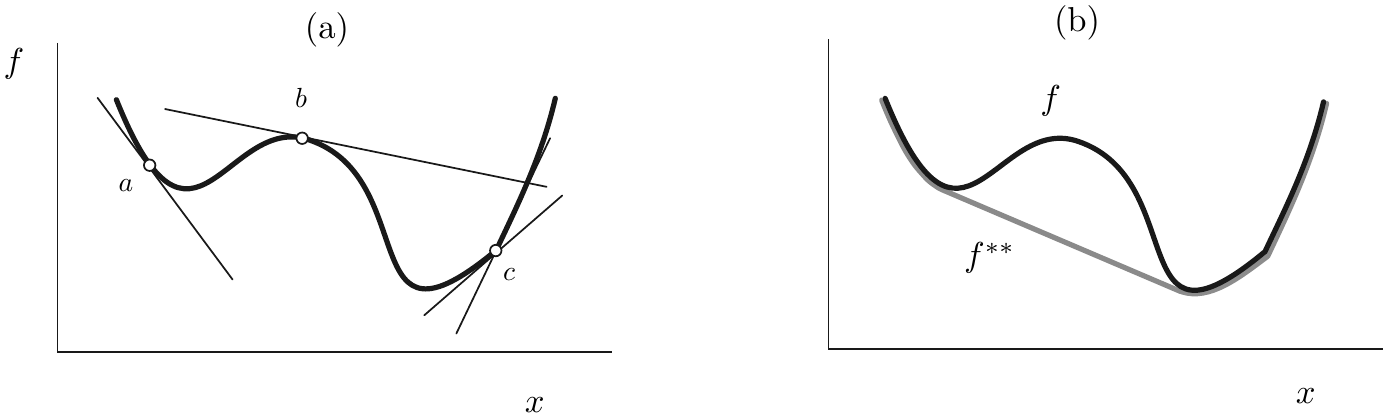}
\caption{(a) Function $f(x)$ with a unique supporting line at the point $a$, no supporting line at the point $b$, and many supporting lines at the point $c$, leading to $\p f(a)=\{f'(a)\}$, $\p f(b)=\emptyset$, and $\p f(c)=[f'(c^-),f'(c^+)]$. (b) Function $f(x)$ and its convex envelope $f^{**}(x)$.}
\label{figconv1}
\end{figure}

\subsection{Legendre--Fenchel transforms}

The \emph{Legendre--Fenchel transform} of $f$ is the real function defined by
\be
f^*(k) = \sup_{x\in\reals} \{kx-f(x)\},\qquad k\in\reals.
\ee
This function is also called the \emph{dual} or \emph{conjugate} of $f$ and has the property of being convex \cite[Thm.~12.2]{rockafellar1970}. The \emph{double dual} or \emph{biconjugate} of $f$ is the Legendre--Fenchel of $f^*$:
\be
f^{**}(x) = \sup_{k\in\reals} \{kx-f^*(k)\}.
\label{eqdoublelf1}
\ee
This is also a convex function, corresponding to the convex envelope or convex hull of $f$ \cite[Thm.~11.1]{rockafellar1988}, as illustrated in Fig.~\ref{figconv1}b. 

With this geometric interpretation of $f^{**}$, it is natural to say that $x$ is a \emph{convex point} of $f$ if $f(x)=f^{**}(x)$ and a \emph{nonconvex point} of $f$ if $f(x)\neq f^{**}(x)$. An important result proved in \cite[Lem.~4.1]{ellis2000} is that the set of convex points coincides with the set of points admitting supporting lines, except possibly at boundary points. With this proviso, we then have $f(x)=f^{**}(x)$ if and only if $\p f(x)\neq \emptyset$. This is illustrated in Fig.~\ref{figconv1}a. The same result also implies that, if $f(x)=f^{**}(x)$, then $\p f(x)=\p f^{**}(x)$.

In this paper, we deal with rate functions, which always have at least one global minimum. Denoting one such minimizer by $x^*$, we then have $0\in \p f(x^*)$. Hence, $x^*$ is a convex point such that $f(x^*)=f^{**}(x^*)$ and $\p f(x^*)=\p f^{**}(x^*)$.

\subsection{Duality}

The proof of our main result, Theorem~\ref{thm1}, is based on another important result about convex functions stating (see \cite[Cor.~23.5.1]{rockafellar1970} or \cite[Prop.~11.3]{rockafellar1988}) that
\be
k\in \p f(x)\iff x\in \p f^*(k).
\label{eqdual1}
\ee
This property expresses a form of duality or conjugacy between the slopes of $f$ and the slopes of $f^*$, illustrated in Fig.~\ref{figduality1}a. From this result, it is easy to see that convex, affine parts of $f$ correspond to cusps of $f^*$, and vice versa, as shown in Fig.~\ref{figduality1}b.

The duality in \eqref{eqdual1} also holds for $f^{**}$, since this function is convex and is the Legendre--Fenchel transform of $f^*$. Therefore,
\be
k\in \p f^{**}(x)\iff x\in \p f^*(k). 
\label{eqdual2}
\ee
This result implies that $f^*$ has a cusp also when $f$ is nonconvex, as shown in Fig.~\ref{figduality1}, since $f^{**}$ is affine where $f$ is nonconvex. Thus, $f^*$ has a cusp either if $f$ is affine or $f$ is nonconvex. 

Since subdifferentials of $f$ and $f^{**}$ match at convex points, it is also clear from \eqref{eqdual2} that the first duality \eqref{eqdual1} holds locally at these points even if $f$ is not globally convex.  We use this result in this paper when dealing with the subdifferential of $I_Q^B$ at its global minimum $w^*$, which is a convex point, as mentioned. In this case, the first duality result can be applied at that point even though $I_Q^B$ might be nonconvex at other points, as in Fig.~\ref{figconvex1}c or Fig.~\ref{figgauss1}.

\begin{figure}[t]
\centering
\includegraphics{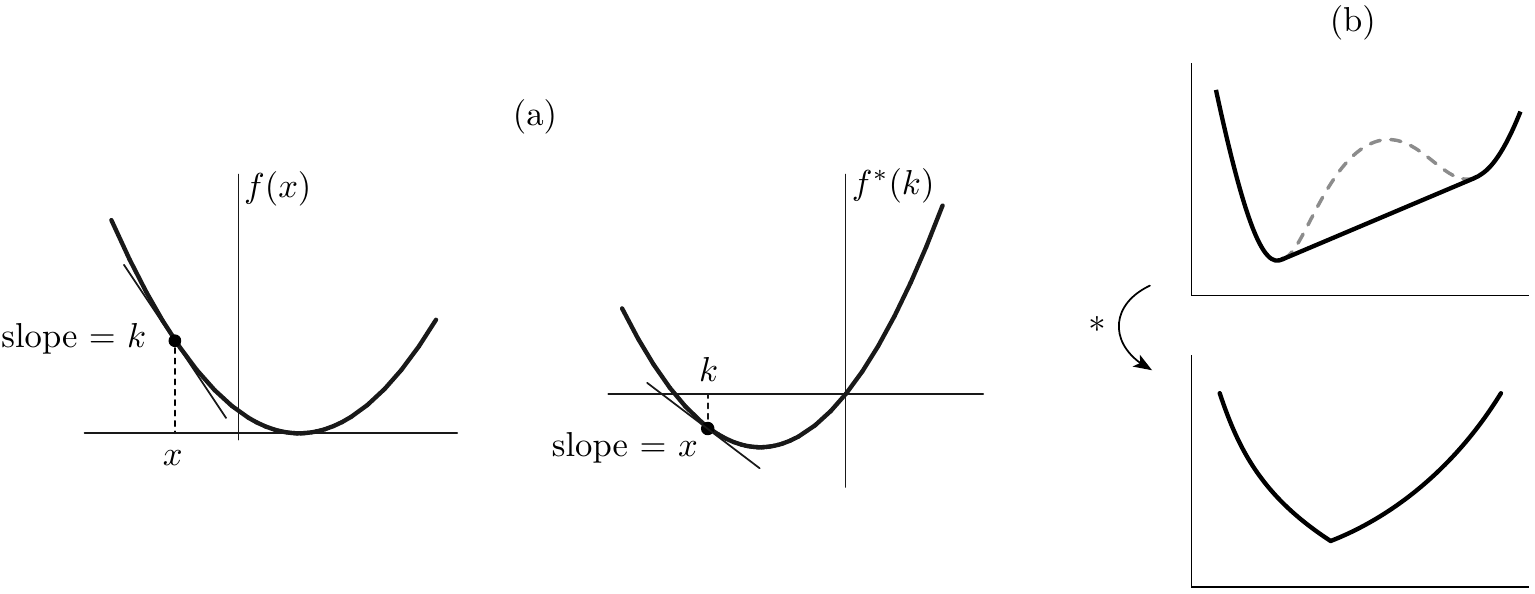}
\caption{(a) Illustration of the duality between the slopes of $f(x)$ and the slopes of its Legendre--Fenchel transform $f^*(k)$. (b) Functions with affine or nonconvex parts give rise to a Legendre--Fenchel transform having a cusp.}
\label{figduality1}
\end{figure}

\section{Contraction principle}
\label{appcont}

The contraction principle is an important result in large deviation theory relating the rate functions of random variables that can be mapped to one another. Let $(A_n)_{n>0}$ be a sequence of random variables satisfying the LDP with good rate function $I_A$ and let $(B_n)_{n>0}$ be another sequence such that $B_n=f(A_n)$ with $f$ continuous. Then $(B_n)_{n>0}$ also satisfies the LDP with good rate function
\be
I_B(b)=\inf_{a:f(a)=b} I_A(a).
\label{eqcont1}
\ee
See \cite[Thm.~4.2.1]{dembo1998} for details.

Instead of considering a single continuous function $f$ as the contraction, one can also consider a sequence $(f_n)_{n>0}$ of continuous functions. In this case, the contraction principle also applies provided that $f_n$ is ``close enough'' to $f$ with respect to $P_n$. To be more precise, let $\cA$ denote the space of $A_n$ and define
\be
\Gamma_{n,\delta}=\{a\in\cA: \|f_n(a)-f(a)\|>\delta\}
\label{eqgammaset1}
\ee
as the set of points for which $f_n$ differs from $f$ by at least $\delta>0$ with respect to any metric $\|\cdot\|$ on $\cB$, the space of $B_n$. Then, according to \cite[Cor.~4.2.21]{dembo1998}, $B_n=f_n(A_n)$ satisfies the LDP with good rate function $I_B$ given by \eqref{eqcont1} with $f$ as the contraction if, for all $\delta>0$,
\be
\lim_{n\ra\infty} \frac{1}{n}\log P_n(\Gamma_{n,\delta})=-\infty.
\label{eqcontcond1}
\ee
This condition only means that the probability that $f_n$ differs from $f$ decreases faster than exponentially with $n$ in the large deviation limit. This is met in most cases when $f_n$ is smooth and $I_A$ is a good rate function.

Two particular applications of this result are considered in the paper. 
\begin{example}
\label{exmultcont1}
Consider two real random variables $A_n$ and $B_n$ related by the simple rescaling $B_n=c_n A_n$ with $c_n\ra 1$ as $n\ra\infty$. Here, the limit function is the identity $f(a)=a$, so one expects $A_n$ and $B_n$ to have the same rate function. This is verified by noting that, for every $M>0$, there exists $n_0=n_0(M,\delta)$ such that for all $n\geq n_0$, one has $\Gamma_{n,\delta}\subseteq(-\infty,-M]\cup[M,\infty)$. Therefore, from the definition of the LDP, we obtain
\be
\limsup_{n\ra\infty} \frac{1}{n}\log P_n(\Gamma_{n,\delta})\leq -\inf_{|a|\geq M} I_A(a).
\label{eqcontcalc1}
\ee
But, since the rate function $I_A$ of $A_n$ is good, it is coercive, so that
\be
\lim_{|a|\ra\infty} I_A(a)=\infty.
\ee
Therefore, the limit on the left-hand side of \eqref{eqcontcalc1} must give $-\infty$, implying $I_B(b) = I_A(b)$ from the condition \eqref{eqcontcond1}.
\end{example}

\begin{example}
\label{exaddcont1}
Let $B_n =f(A_n)+c_n$ with $c_n\ra c$. Then the rate function of $B_n$ is obtained from \eqref{eqcont1} with the contraction $B_n=f(A_n)+c$. This follows trivially because the distance between $f(a)+c_n$ and $f(a)+c$ is constant in $a$. Since $c_n\ra c$, there must be an $n$ beyond which $|c_n-c|<\delta$, leading to $P_n(\Gamma_{n,\delta})=0$, so the condition \eqref{eqcontcond1} is also satisfied.
\end{example}

These results also hold if $\Gamma_{n,\delta}$ is defined on a subset of $\cA$, since any restriction or constraint on $A_n$ can be included in the definition of $f_n$. This arises, for example, when considering the contraction of $J_Q(m,w)$ to $I_Q^B(w)$, which involves the restriction $m\in B$. 

\begin{acknowledgments}
We are greatly indebted to Julien Reygner for valuable comments and insightful suggestions on the first version of the paper, which led to some technical modifications in Assumption~\ref{hyp2}, Assumption~\ref{hyp3}, and Eq.~\eqref{eqdefiqb} in this version. We also thank Gr\'egoire Ferr\'e and Gabriel Stoltz for carefully reading the paper. A.G. thanks Maxime Sangnier for fruitful discussions during the writing of this paper. H.T.\ is supported by Stellenbosch University (Establishment Funds) and the National Research Foundation of South Africa (Grant No.\ 96199).
\end{acknowledgments}

\bibliography{masterbibmin}

\end{document}